\documentclass[11pt]{article}
\usepackage[hyphens]{url}
\usepackage{graphicx}
\usepackage[round]{natbib}
\usepackage{caption}
\usepackage[margin=1in]{geometry}
\usepackage[T1]{fontenc}
\usepackage{algorithm}
\usepackage[noend]{algorithmic}

\usepackage{xcolor}
\usepackage[most]{tcolorbox}

\newcommand{\dpmap}[1]{\ensuremath{\mathsf{feasible}^{(#1)}}}

\newcommand{\hcell}[1]{\colorbox{yellow!35}{$#1$}}

\newcommand{\pp}{\begin{matrix}0\\1\\\vec{1}\end{matrix}}
\newcommand{\qq}{\begin{matrix}1\\0\\\vec{1}\end{matrix}}

\newcommand{\pptop}{\begin{matrix}\hcell{0}\\1\\\vec{1}\end{matrix}}
\newcommand{\ppmid}{\begin{matrix}0\\\hcell{1}\\\vec{1}\end{matrix}}
\newcommand{\qqtop}{\begin{matrix}\hcell{1}\\0\\\vec{1}\end{matrix}}
\newcommand{\qqmid}{\begin{matrix}1\\\hcell{0}\\\vec{1}\end{matrix}}

\usepackage{newfloat}
\usepackage{listings}
\DeclareCaptionStyle{ruled}{labelfont=normalfont,labelsep=colon,strut=off}
\floatstyle{ruled}
\newfloat{listing}{tb}{lst}{}
\floatname{listing}{Listing}

\usepackage{booktabs}

\usepackage{amsmath}
\usepackage{amsthm}
\usepackage{amsfonts}
\usepackage{hyperref}
\usepackage[svgnames]{xcolor}
\hypersetup{colorlinks={true},urlcolor={blue},linkcolor={DarkBlue},citecolor=[named]{DarkGreen}}

\usepackage[tt=false]{libertine}
\usepackage[libertine]{newtxmath}

\usepackage[capitalise,nameinlink,noabbrev]{cleveref}
\usepackage{xspace}

\newtheorem{lemma}{Lemma}
\newtheorem{theorem}{Theorem}
\newtheorem{observation}{Observation}
\newtheorem{corollary}{Corollary}
  
\theoremstyle{definition}
\newtheorem{definition}{Definition}

\theoremstyle{remark}

\newtheoremstyle{problemstyle}
  {\topsep}
  {\topsep}
  {\normalfont}
  {0pt}
  {\bfseries}
  {.}
  {\newline}
  {\thmname{#1}\thmnumber{ #2}\thmnote{ (#3)}\rule[-1.5ex]{0pt}{0pt}}

\theoremstyle{problemstyle}
\newtheorem{problem}{Problem}

\crefname{observation}{Observation}{Observations}
\Crefname{observation}{Observation}{Observations}
\crefname{question}{Question}{Questions}
\Crefname{question}{Question}{Questions}
\crefname{problem}{Problem}{Problems}
\Crefname{problem}{Problem}{Problems}

\title{Binary $k$-Center under a~Hard Threshold with~Applications to~Delegated Voting\thanks{The majority of this work was performed while Aris Filos-Ratsikas was at the University of Edinburgh.}}
\author{
\begin{tabular}{c}
Jakub Dargaj\textsuperscript{1}\quad
Aris Filos-Ratsikas\textsuperscript{2}\quad
Paul W. Goldberg\textsuperscript{3}\\[0.5em]
\textsuperscript{1} University of Edinburgh, United Kingdom\\
\textsuperscript{2} Imperial College London, United Kingdom\\
\textsuperscript{3} University of Oxford, United Kingdom
\end{tabular}
}
\date{}

\begin{document}
\maketitle

\begin{abstract}
We study the problem of covering binary strings of the same length, possibly with missing entries, by a fixed number of center strings, such that each input string is within a given relative distance of the closest center. The problem has applications in binary $k$-center clustering and bioinformatics, but our main motivation comes from computational social choice. In the setting of multi-issue approval voting, we consider an algorithmic question that precedes any election — how should representatives be designed so that as many voters as possible are willing to delegate? We study the problem in two dimensions, namely the number of centers and the agreement threshold, parameters that are application-specific, and provide a complete picture of its computational complexity when entries may be missing. For the special case without missing entries, we extend our hardness results to any number of centers and large values of threshold, leaving a narrow range of thresholds unresolved.
\end{abstract}

\section{Introduction}

We study the following fundamental covering problem. We are given a set of $n$ binary strings of length $m$, possibly with missing entries. The task is to construct $k$ complete \emph{center} strings (or, simply, \emph{centers}) in $\{0,1\}^m$ which ``represent'' as many of the $n$ strings as possible. A center represents a string if the two strings agree on at least a $\tau$-fraction of its non-missing coordinates, for some appropriate chosen parameter $\tau$, called the \emph{threshold}. A string is represented by a set of $k$ centers, if it is represented by at least one center in the set. 

We refer to this problem as the \emph{binary $k$-center problem under a hard threshold ($(k,\tau)$-\textsc{center})}. Indeed, the problem can be seen as a version of the well-known $k$-center problem (e.g., see \citep{vazirani2001approximation}) equipped with a Hamming distance function, but with a different objective: rather than minimizing the maximum distance between any point (string) and its closest center, we would like to make sure that as many points as possible (or, ideally, all points) are within a given Hamming distance from their closest center, with this distance being determined by the threshold $\tau$.   

This problem has a host of potential applications, and we highlight some of these below.

\paragraph{Delegated Voting.} In representative democracy, voters often delegate their votes to delegates with whom they share a consensus of opinions regarding some, but not necessarily all issues. Concretely, in this application each string is a voter who has \emph{approval preferences} over a set of issues (e.g., signifying ``in favor'' or ``against''), and the missing entries correspond to issues which the voter is not informed about or which the voter simply does not care about. The threshold $\tau$ captures the \emph{agreement threshold}, over which the voter would be comfortable delegating her vote to a given delegate. The goal is thus to come up with $k$ delegates that collectively represent, or \emph{attract}, the largest possible number of voters.

Proxy voting with incomplete votes has notable applications in blockchain governance, as in e.g., Cardano's Project Catalyst\footnote{See \url{https://projectcatalyst.io/} and \url{https://www.1694.io/}.}, in which voters can evaluate only a small subset of the many community-generated proposals. Similarly, well-known voting advice applications such as the German \emph{Wahl-O-Mat} or the Swiss \emph{smartvote} \citep{garzia_matching_2014} compute agreement scores between voters and parties over a range of issues. The construction of $k$ delegates can thus be interpreted as the choice of $k$ members of a political party which collectively ``attract'' the largest possible number of voters within a given electorate.

\paragraph{DNA Sequencing.} \emph{Haplotype assembly} is the problem of reconstructing the distinct copies of a chromosome that an organism inherited from its parents. The input is a set of sequencing \emph{reads}, each covering only a small, error-prone fragment of the genome, and the output is one consensus sequence per inherited copy, see \citep{patterson_whatshap_2015}. Concretely, each read is a row of a binary \emph{fragment matrix} with missing entries (i.e., our $n$ input strings), and the output is a set of $k$ complete binary strings, one per chromosome copy (i.e., our centers). The value of $k$ is determined by the application: for diploid organisms such as humans we have that $k = 2$, while polyploid genomes mostly come from plants; for example, the potato has $k = 4$ and wheat has $k = 6$. This problem has been studied under the \emph{minimum error correction} (MEC) objective, i.e., the fewest bit flips making the matrix consistent with $k$ haplotypes; which corresponds to the standard $k$-median objective under Hamming distance. In contrast, $(k,\tau)$-\textsc{center} is concerned with whether the matrix can be made conflict-free assuming a bounded error rate per read. This is quite meaningful for long reads, where MEC may result in incorrectly reconstructed haplotypes \citep{majidian_minimum_2020}, and sequencing errors are approximately uniformly distributed along the read \citep{dohm_benchmarking_2020}. The column-wise counterpart, bounding the number of errors per position, has been studied by \citet{pirola_hapcol_2016} and \citet{beretta_hapchat_2018}.

\paragraph{AI User Modeling.} Our setting could also serve as an abstraction of the very recent application of user modeling via AI models. In this regime, an AI agent aims to learn to represent or even ``simulate'' a human, by eliciting the opinions of the human on a chosen set of issues, e.g., see \citep{lee_webuildai_2019,jarrett_language_2023,fish_generative_2024}; in our setting these opinions would be captured by the binary strings with missing entries. Instead of having a personalized AI agent for each individual, the goal is often to come up with a small set of models (our $k$ centers) that accurately represent as many members of society as possible (e.g., see \citep{sorensen2024roadmap} and the notion of \emph{pluralistic alignment}). \medskip 

\noindent Given that $(k,\tau)$-\textsc{center} serves as an abstraction for a plethora of important applications, we are interested in establishing what is theoretically possible when looking for good solutions to the problem. Concretely, we study its \emph{computational complexity}. Instead of treating the values of $k$ and $\tau$ as part of the input (which would make the problem easily seen to be NP-hard in general), we classify the computational complexity for every fixed pair of $k$ and $\tau$, which are interpreted as constant parameters associated with the problem. We provide a more formal definition of the problem, as well as an overview of our main results below.

\subsection{Our Contributions}

We study the computational complexity of the following decision problem, where $k$ and $\tau$ are constant parameters known in advance: $k$ is a positive integer representing the number of centers, and $\tau \in [0,1]$ is a rational number representing the minimum agreement threshold between a string and the corresponding center. The input strings are over the binary alphabet, with missing entries represented by the asterisk ($\star$), and the Hamming distance $d_H$ is applied only over non-missing entries, that is, $$d_H(x, y) = \sum_{j \in [m], x_j \neq \star}|x_j - y_j|.$$

\smallskip 

\begin{tcolorbox}[
    colback=gray!8,      
    colframe=gray!50,    
    boxrule=0.6pt,       
    arc=3mm,             
    left=4mm,
    right=4mm,
    top=2mm,
    bottom=2mm
]
\begin{problem}[$(k,\tau)$-\textsc{center}]
\label{prob:ktc}

\noindent \emph{\textbf{Input}:} $n$ strings $x_1,\dots,x_n \in \{0,1,\star\}^m$. \medskip 

\noindent\emph{\textbf{Output:}} \textbf{Yes}, if there is a set of $k$ \emph{center} strings
$y_1,\dots,y_k \in \{0,1\}^m$ such that
\[
  \min_{j \in [k]} d_H(x_i, y_j) \le (1-\tau)\,m_i
  \quad \text{for all } i \in [n]\,,
\]
where $m_i = |\{j \in [m]: x_{i,j} \neq \star \}|$ is the number of coordinates present at $x_i$, and \textbf{No} otherwise. If the answer is \textbf{Yes}, also return the set of $k$ center strings.
\end{problem}
\end{tcolorbox}

Our main contribution is the classification of the tractability of the problem in two dimensions: the number of centers $k$, and the agreement threshold $\tau$. For our main set of results, we provide a complete picture of tractability when the strings are incomplete, i.e., when they contain missing entries; these are summarized in \cref{table:general-classification}.

\begin{table}[t]
\centering
\renewcommand{\arraystretch}{1.25}
\setlength{\tabcolsep}{5pt}

\begin{tabular}{c|cccc}
 & $\tau = 0$ & $\tau \in (0,\frac12]$ & $\tau \in (\frac12,1)$ & $\tau = 1$ \\
\hline
$k=1$    & P & NP-hard & NP-hard & P \\
$k=2$    & P & P       & NP-hard & P \\
$k\geq3$ & P & P       & NP-hard & NP-hard \\
\end{tabular}
\caption{Complexity of $(k,\tau)$-\textsc{center} classified by $k$ and $\tau$.}
\label{table:general-classification}
\end{table}

From a technical perspective, the positive results (membership in $P$) are easy to obtain; for example, it is not hard to see that for $k \geq 2$ and $\tau \leq 1/2$, any input string will be within the agreement distance from either any center or its complement. The computational hardness results are more involved and are established via the appropriate application of \emph{dichotomy theorems} for \emph{constraint satisfaction problems (CSPs)}, see \citep{schaefer_complexity_1978,chen_rendezvous_2009,barto_polymorphisms_2017}; we provide the necessary technical background for these concepts before the exposition of our results.  

We also consider the interesting special case of $(k,\tau)$-\textsc{center} with \emph{complete} strings, in which no entries are missing; that is, $x_i \in \{0,1\}^m$ for each $i \in [n]$. For this case the landscape of intractability that we obtain is nearly complete, but leaves an unresolved range of thresholds for future work. We present the results for the case of complete strings in the theorem below. 

\begin{theorem}
\label{thm:complete-summary}
$(k,\tau)$-\textsc{center} with complete strings is:
\begin{itemize}
    \item[-] in P for $k=1$ and $\tau < \frac12$, for $k \geq 2$ and $\tau \leq \frac12$, and whenever $\tau=1$,
    \item[-] NP-complete for $k=1$ and $\frac12 \leq \tau < 1$,
    \item[-] NP-complete for $k=2$ and $\frac23 < \tau < 1$,
    \item[-] NP-complete for any $k \geq 3$ and $\tau_k^\star < \tau < 1$, where
    $$\tau_k^\star = \frac{5\lceil \frac{k-1}{2} \rceil - 2}{6\lceil \frac{k-1}{2} \rceil - 2} < \frac56.$$
\end{itemize}
\end{theorem}

\noindent We refer to the maximization version of the problem, where the goal is to find a set of $k$ centers such that the maximum possible number of input strings is covered, as \textsc{max}-$(k,\tau)$-\textsc{center}. Our NP-hardness results clearly carry over to the maximization version as well, but, in fact, similar dichotomy theorems can be used to show stronger impossibility (inapproximability) results; see the Appendix for more details.
Additionally, approximation algorithms or parameterized complexity results for the problem are beyond the scope of this work, but we offer an interesting related discussion in the Discussion and Future Work section, and some preliminary results in the Appendix.

\subsection{Discussion and Related Work}

Given the rather fundamental nature of the $(k,\tau)$-\textsc{Center} problem, to the interested reader it might look like a problem that should have been studied before. There is in fact a number of related string covering and clustering problems that have been studied in the literature, but, to the best of our knowledge, none of those is identical to $(k,\tau)$-\textsc{Center}. The main differentiating factor is the hard threshold, and the corresponding objective of maximum covering under this threshold.  

The complementary objective, i.e., minimizing the maximum Hamming distance between an input string and its closest center, has been studied under the names \emph{$k$-center clustering} and \emph{Hamming radius $c$-clustering} \citep{gasieniec_approximation_2004, amir_efficiency_2014, knop_combinatorial_2020, friedrich_binary_2025}, where $k$ and $c$ denote the number of clusters. The problem is NP-hard already for a single cluster and without missing entries, as it is equivalent to \emph{Closest String} (also known as \emph{Consensus String} or \emph{the Hamming center problem}). Its NP-hardness is due to \citet{frances_covering_1997}, and it is known to admit a PTAS, see \citep{li_closest_2002,andoni_optimality_2006,ma_more_2010}.

Perhaps most related to our problem is the problem \emph{Close to Most Strings}, see \citep{boucher_approximating_2013}. In this problem, we are looking for a single binary string within a given Hamming distance of the maximum possible number of input strings. This corresponds to a rather restricted variant of our problem where (a) $k=1$, (b) all input strings are complete (i.e., without missing entries) and \emph{crucially} (c) $\tau$ is an input to the problem. Given this, \citet{boucher_approximating_2013} used $\tau = 1/2$ to show that the problem does not admit a PTAS unless ZPP = NP.\footnote{The NP-hardness of the problem follows from the NP-hardness of Closest String, as the decision versions of the two problems coincide. \citet{ma_polynomial_2000} stated an APX-hardness result, but their proof was found erroneous subsequently by \citet{boucher_approximating_2013}.} Our problem can thus be seen as a broad generalization of Close to Most Strings where the strings may have missing entries, and $k$ and $\tau$ can have multiple values and are treated as constants rather than input parameters. 

In the context of delegated voting, \citet{amanatidis_potential_2026} used a very similar model of threshold agreement, as the one used in $(k,\tau)$-\textsc{Center}, focusing on the case of $\tau=1/2$, which they coin ``Majority Agreement''. Importantly, their objective is markedly different: they are merely interested in the outcome of the voting process rather than the issue of representation, namely, how should $k$ delegates cast their votes so that the outcome of the voting process between the voters they represent (and those voters that are not represented by any delegate and vote directly) has as large as possible approval score. In particular, it follows from the results of \citet{amanatidis_potential_2026} that a high level of representation is not necessarily aligned with good voting outcomes. We refer the reader to their work for a more detailed discussion on related works in proxy voting. 

Finally, a single delegate can be viewed as a committee over the issues, making $(1,\tau)$-delegation analogous to approval-based committee selection. To this end, the bulk of the related literature assumes a fixed committee size, which corresponds to delegates with a fixed Hamming weight; see the recent survey of \citet{lackner_multiwinner_2023}. General threshold rules, in which a voter is satisfied by a committee containing enough of their approved candidates, were introduced by \citet{fishburn_approval_2004}. \citet{faliszewski_multiwinner_2020} study specific thresholds for variable-sized committees, closely matching our setting, and show that it is NP-hard to decide whether a non-empty committee satisfies all voters under the majority threshold rule. Their threshold, however, is normalized by the committee size and counts only approvals inside the committee, whereas ours is normalized by the voter's ballot and computes agreement on both approved and disapproved issues.

\section{Technical Prerequisites}

This section reviews the basic terminology and tools from coding theory and constraint satisfaction problems used throughout the paper. At a high level, we use dichotomy theorems for CSPs to obtain our hardness results when strings have missing entries, encoding each input string as a constraint on the centers, and we use codes of large minimum distance to construct instances in which no two clusters can share a center, in our hardness results for complete strings. For CSPs, and in particular for the algebraic formulation of dichotomy theorems, we follow the notation and conventions of \citet{chen_rendezvous_2009}, while for coding theory background we refer to the classical texts of \citet{vanlint_introduction_1999} and \citet{macwilliams_theory_1977}.

\subsection{Constraint Satisfaction Problems}

Let $\Gamma$ be a \emph{constraint language}, which is a set of relations over a finite domain $D$. For a given finite set of variables $U$, a \emph{constraint} $c$ over $\Gamma$ is an expression $R(v_1,\dots,v_k),$ where $R \in \Gamma$ is a relation of arity $k$ and $v_i \in U$ are distinct variables from $D$. An \emph{assignment} is a mapping $f: U \rightarrow D$, and we say that $f$ \emph{satisfies} $R(v_1,\dots,v_k)$ if $(f(v_1),\dots,f(v_k)) \in R$. The problem CSP($\Gamma$) takes a finite set of variables $U = \{v_1,\dots,v_m\}$ and a finite set of constraints $C=\{c_1,\dots,c_n\}$ over $\Gamma$ as input, and decides whether there exists a \emph{satisfying assignment} $f$ -- an assignment satisfying all constraints in $C$.

A function $g: \{0,1\}^m \rightarrow \{0, 1\}$ is a \emph{polymorphism} of a $k$-ary relation $R$ if for any $m$ tuples $v^{(1)},\dots,v^{(m)}$ from $R,$ also the tuple $\left(g(v^{(1)}_1,\dots,v^{(m)}_1), \dots, g(v^{(1)}_k,\dots,v^{(m)}_k)\right)$ is in $R$. We say that $g$ is a polymorphism of a constraint language $\Gamma$ if $g$ is a polymorphism of all relations in $\Gamma$.

\subsubsection{CSPs over Binary Domain}

Algebraic formulation of Schaefer's theorem below characterizes tractable problems CSP($\Gamma$) in terms of polymorphisms of $\Gamma$, when restricted to finite constraint languages with relations over binary domain.

\begin{theorem}[\citet{schaefer_complexity_1978,chen_rendezvous_2009}]
CSP$(\Gamma)$ over binary domain is polynomial-time tractable if $\Gamma$ has any of the following six functions as a polymorphism:
\begin{itemize}
    \item[-] unary constant functions $\mathbf{0}$ and $\mathbf{1}$,
    \item[-] binary functions $\mathbf{AND}$ and $\mathbf{OR}$,
    \item[-] ternary functions $\mathbf{majority}$ and $\mathbf{parity}$.
\end{itemize}
Otherwise, CSP$(\Gamma)$ is NP-complete.
\end{theorem}

\subsubsection{CSPs over Finite Domains}

The more general CSP dichotomy conjecture over any finite domain, which generalizes Schaeffer's dichotomy for the Boolean domain, was formulated by \citet{feder_computational_1998} and only proved recently by \citet{zhuk_proof_2020} and \citet{bulatov_dichotomy_2017}. We use the equivalent formulation due to \citet{barto_polymorphisms_2017}, stated in terms of the existence of a quaternary Siggers polymorphism. 

\begin{definition}
Let $\Gamma$ be a finite constraint language over a finite domain $D$.
A $4$-ary polymorphism $f: D^4 \rightarrow D$ is called a \emph{Siggers polymorphism} if for all $a,r,e \in D:$
$$
f(a,r,e,a)=f(r,a,r,e).
$$
\end{definition}

\begin{observation}
\label{observation:siggers-xy}
If $f$ is a Siggers polymorphism, then for all $x,y \in D$ we have
$$
f(x,y,x,x)=f(y,x,y,x)=f(y,x,x,y).
$$
\end{observation}
\begin{proof}
Substituting $(a,r,e)=(x,y,x)$ gives the first equality and substituting $(a,r,e)=(y,x,x)$ gives the second equality.
\end{proof}

\begin{theorem}[\citet{barto_polymorphisms_2017}, Corollary 42]
\label{theorem:barto-siggers-np-hard}
If a finite constraint language $\Gamma$ has no Siggers polymorphism then CSP$(\Gamma)$ is NP-complete.
\end{theorem}

\subsection{Error-Correcting Codes}

A \emph{binary code} $\mathcal{C} \subseteq \{0,1\}^m$ of length $m$ is any set of binary strings of length $m$, called \emph{codewords}. The \emph{minimum distance} $d_\mathcal{C}$ of $\mathcal{C}$ is defined as the minimum Hamming distance $d_H(x, y)$ over all distinct codewords $x, y \in \mathcal{C}.$ For a binary vector $x \in \{0,1\}^m$ and $0 \leq r \leq m,$ we denote $B(x,r)$ the \emph{Hamming ball} of radius $r$ around $x$, that is, the set of binary vectors with Hamming distance at most $r$ from $x$; the Hamming weight of $x$ is the number of ones in $x$.

The binary entropy function $H: [0,1] \to [0,1]$ is defined as $H(x) = -x \log_2 x - (1-x) \log_2(1-x)$ for $x \in (0, 1)$ and $H(0)=H(1)=0$. As the size of the $m$-dimensional Hamming ball of radius $r$ around $x$ is $|B(x,r)|=\sum_{i=0}^{r}{\binom{m}{i}}$, the binary entropy helps to estimate $|B(x,r)|$ using the well-known inequality $\sum_{i=0}^{\alpha m}{\binom{m}{i}} \leq 2^{mH(\alpha)}$ for any $0 < \alpha < \frac12$.

We will use the following three classical bounds from coding theory.

\begin{theorem}[Plotkin bound]
\label{theorem:plotkin}
    Any binary code $\mathcal{C}$ of length $m$ and minimum distance $d > \frac{m}{2}$ satisfies
    $$|\mathcal{C}| \leq \frac{2d}{2d-m}.$$
\end{theorem}

\begin{theorem}[Gilbert-Varshamov bound]
\label{theorem:g-v}
    For any two positive integers $m$ and $d$ such that $d\leq m$, there exists a binary code $\mathcal{C}$ of length $m$ and minimum distance $d$ satisfying
    $$|\mathcal{C}| \geq \frac{2^m}{\sum_{i=0}^{d-1}{\binom{m}{i}}}.$$
    Moreover, for $d \leq \frac{m}{2}$ it holds that
    $$|\mathcal{C}| \geq 2^{m(1-H((d-1)/m))},$$
    where $H$ is the binary entropy function.
\end{theorem}

\begin{theorem}[Johnson bound]
\label{theorem:johnson}
    Let $\mathcal{C}$ be any binary code of length $m$ and minimum distance $d_\mathcal{C}=\frac{1}{2}(1-\delta)m$ for some $0 < \delta < 1$. Let $d_\mathrm{max}=\frac{1}{2}(1-\gamma)m$ for some $0 < \gamma < 1$ and assume $\gamma > \sqrt{\delta}$. Then for any $x \in \{0,1\}^m$,
    $$|B(x,d_\mathrm{max}) \cap \mathcal{C}| \leq \min(m, \frac{1-\delta}{\gamma^2-\delta}).$$
\end{theorem}

\section{Complexity of \texorpdfstring{$(k,\tau)$}{(k, τ)}-\textsc{center}}

We now establish our main result: the complexity of $(k,\tau)$-\textsc{center}, classified by $k$ and $\tau$, summarized in \cref{table:general-classification}. Membership in NP is immediate for any $k \geq 1$ and $\tau \in [0,1]$: given center strings $y_1,\dots,y_k$, we can easily check that they together cover all the input strings. For a single center, we show NP-hardness by a reduction from a binary CSP; for multiple centers, we reduce from a CSP over domain $\{0,1\}^k$.

\subsection{Single Center}

We show NP-hardness of $(1,\tau)$-\textsc{center} for any parameter $\tau \in (0,1)$ by a polynomial-time reduction from the binary constraint satisfaction problem $\mathrm{CSP}(\Gamma_{\tau, M})$, where $M$ depends on $\tau$ and the language $\Gamma_{\tau, M}$ is defined below. In fact, we show hardness for a restricted class of $(1,\tau)$-\textsc{center} instances in which every string has exactly $M$ revealed positions, and its revealed entries are either all $1$ or all $0$. Let us first define necessary binary vectors and relations.

\begin{definition}
\label{definition:string-relation}
Let $\tau \in (0, 1)$ and $M$ be a positive integer. For each $i \in {0, \dots, M}$, we define
\begin{itemize}
    \item[-] a binary vector $s^{(M,i)} \in \{0,1\}^M$ of length $M$ such that $s^{(M,i)}_j = 1$ for $j \leq i$ and $s^{(M,i)}_j=0$ for $j > i$,
    \item[-] a relation $R^{(\tau,M,i)}\subseteq \{0,1\}^M$ such that $v \in R^{(\tau,M,i)}$ iff $d_H(v, s^{(M,i)}) \leq (1-\tau)M$.
\end{itemize}
\end{definition}

\begin{definition}
\label{def:gamma-tau-m}
Let $\tau \in (0, 1)$ and $M$ be a positive integer. We define the constraint language
$$\Gamma_{\tau, M} = \{R^{(\tau,M,0)},R^{(\tau,M,M)}\}.$$
\end{definition}

When $\tau$ and $M$ are clear from context, we denote the relations $R^\textbf{no}=R^{(\tau,M,0)}$ and $R^\textbf{yes}=R^{(\tau,M,M)}$ for simplicity. We obtain NP-hardness of $\mathrm{CSP}(\Gamma_{\tau, M})$ by showing that $\Gamma_{\tau, M}$ has none of the six polymorphisms from Schaefer's theorem.

\begin{lemma}
\label{lemma-csp-hard}
$\mathrm{CSP}(\Gamma_{\tau, M})$ is NP-hard whenever $\tau \in (0, 1)$ and $M \geq 2/\min(\tau,1-\tau)$.
\end{lemma}

\begin{proof}
    We show that none of the six functions in Schaefer's theorem is a polymorphism of $\Gamma_{\tau, M}$.

    First, observe that for $M \geq 1/\tau$, $(1,\dots,1) \notin R^\textbf{no}$, hence $\mathbf{1}$ is not a polymorphism of $R^\textbf{no}$. Similarly, $\mathbf{0}$ is not a polymorphism of $R^\textbf{yes}$, which implies that none of the constant operations $\mathbf{0}$ and $\mathbf{1}$ is a polymorphism of $\Gamma_{\tau,M}$.

    For the binary operations, consider the relation $R^\textbf{yes}$, the minimum agreement $\mu = \lceil \tau M \rceil$, and assume $M \geq 1/\min(\tau,1-\tau)$ so that $1 \leq \mu \leq M-1$. Take the binary vector $s^{(M,\mu)}$ and its reverse $r$, and observe that $s^{(M,\mu)}, r \in R^\textbf{yes}$ and $s^{(M,\mu)} \neq r$. Applying $\mathbf{AND}$ to $s^{(M,\mu)}$ and $r$ element-wise, we obtain a vector $q$ whose first $M-\mu$ and last $M-\mu$ elements are equal to 0, and the Hamming distance $d_H(q, s^{(M,\mu)}) = \min(M, 2(M-\mu)) > (1-\tau)M$. Therefore, $q \notin R^\textbf{yes}$ and $\mathbf{AND}$ is not a polymorphism of $R^\textbf{yes}$. Similarly, $\mathbf{OR}$ is not a polymorphism of $R^\textbf{no}$, and none of the two operations is a polymorphism of $\Gamma_{\tau,M}$.

    For the operation $\mathbf{majority}$, consider again the relation $R^\textbf{yes}$, the minimum agreement $\mu = \lceil \tau M \rceil$, but now assume $M \geq 2/\min(\tau,1-\tau)$ so that $2 \leq \mu \leq M-2$. Take the binary vector $s^{(M,\mu)}$ and two vectors $r$ and $q$ obtained from $s^{(M,\mu)}$ by shifting the last element with value 1 to the right, that is:
    \begin{itemize}
        \item[-] $r_i=1$ for $i \in \{1,\dots,\mu-1,\mu+1\}$ and $r_i=0$ for $i \in \{\mu, \mu+2, \mu+3,\dots,M\}$,
        \item[-] $q_i=1$ for $i \in \{1,\dots,\mu-1,\mu+2\}$ and $q_i=0$ for $i \in \{\mu, \mu+1, \mu+3,\mu+4,\dots,M\}$.
    \end{itemize}
    Applying $\mathbf{majority}$ to $s^{(M,\mu)}$, $r$ and $q$ element-wise, we obtain the vector $s^{(M,\mu-1)}$, which is not in $R^\textbf{yes}$ by the definition of $\mu$. Therefore, $\mathbf{majority}$ is not a polymorphism of $R^\textbf{yes}$.

    Last, for the operation $\mathbf{parity}$, consider again the relation $R^\textbf{yes}$, the minimum agreement $\mu = \lceil \tau M \rceil$, and assume $M \geq 1/\tau$ so that $\mu \geq 1$. Take the binary vector $s^{(M,\mu)}$ and its reverse $r$ so that both $s^{(M,\mu)}, r \in R^\textbf{yes}$. If $\mu > \frac{M}{2}$, applying $\mathbf{parity}$ to $s^{(M,\mu)}$, $r$ and $s^{(M,M)}$ element-wise is the same as applying $\mathbf{AND}$ to $s^{(M,\mu)}$ and $r$, hence we obtain a binary vector $q \notin R^\textbf{yes}$. Otherwise, $\mu \leq \frac{M}{2}$, let $q$ be a binary vector obtained by applying $\mathbf{OR}$ to $s^{(M,\mu)}$ and $r$. Applying $\mathbf{parity}$ to $s^{(M,\mu)}$, $r$ and $q$ element-wise we obtain the binary vector $s^{(M,0)} \notin R^\textbf{yes}$. In both cases, $\mathbf{parity}$ is not a polymorphism of $R^\textbf{yes}$.
\end{proof}

\begin{theorem}
\label{theorem:delegation-np-complete}
$(1,\tau)$-\textsc{center} is NP-complete for any $\tau \in (0, 1)$.
\end{theorem}
\begin{proof}
Let $M=\lceil2/\min(\tau,1-\tau)\rceil$ so that $\mathrm{CSP}(\Gamma_{\tau, M})$ is NP-hard. Given an instance of $\mathrm{CSP}(\Gamma_{\tau, M})$ with the set of variables $U=\{v_1,\dots,v_m\}$ of size $m \geq M$, and the set of constraints $C=\{c_1,\dots,c_n\}$, we construct an instance of $(1,\tau)$-\textsc{center} as follows.

The strings have length $m$, with one position for each variable $v_j \in U$, and there are $n$ strings, one for each constraint $c_i \in C$. Each string $x_i$, $i \in [n]$, has exactly $M$ revealed positions, so that $m_i = M$. If the constraint $c_i$ is of type $R^\textbf{no}(v_{j_1},\dots,v_{j_M})$ for some distinct indices  $j_1,\dots,j_M \in [m]$, then $x_i$ has value $0$ at each of the positions $j_1,\dots,j_M$ and $\star$ elsewhere. Otherwise, the constraint $c_i$ is of type $R^\textbf{yes}(v_{j_1},\dots,v_{j_M})$, and $x_i$ has value $1$ at each of the positions $j_1,\dots,j_M$ and $\star$ elsewhere.

Observe that there is a 1-to-1 correspondence between any assignment $f: U \rightarrow \{0,1\}$ of $\mathrm{CSP}(\Gamma_{\tau, M})$, and the center $y\in\{0,1\}^m$ such that $y_j=f(v_j)$, in a sense that $f$ satisfies exactly the constraints whose corresponding strings are covered by $y$. This implies that a satisfying assignment of $\mathrm{CSP}(\Gamma_{\tau, M})$ exists if and only if there is a single center that covers all strings.
\end{proof}

\subsection{Multiple Centers}

We now study $(k,\tau)$-\textsc{center} for $k \geq 2$ and $\tau > \frac12,$ as for $\tau \leq \frac12$ the answer is always positive: the centers $0^m$ and $1^m$ cover all strings, since every string agrees with one of them on at least half of its revealed positions. For $\tau > \frac12$, we turn to constraint satisfaction problems over a more general domain.

\begin{definition}
Let $\tau \in (\frac12, 1)$, let $M$ and $k$ be positive integers, and let $w = \lfloor\frac{1}{2}(1-\tau)M\rfloor + 1$. We define a relation $S^{(\tau,M,k)}\subseteq \left\{\{0,1\}^k\right\}^M$ such that $X \in S^{(\tau,M,k)}$ iff $d_H(x_i, 1^{M-2w}0^{2w}) \leq (1-\tau)M$ for some row $x_i = (X_{i,1}, \dots, X_{i,M})$ of $X$.
\end{definition}

\begin{observation}
\label{observation:rep-not-rel}
If $M > \frac{2}{2\tau-1}$, then for any $s \in \{0,1\}^k$, $s^M \notin S^{(\tau,M,k)}.$
\end{observation}
\begin{proof}
    Any row $x_i$ of $s^M$ is either $0^M$ or $1^M$. In the first case, $d_H(0^M, 1^{M-2w}0^{2w}) = M-2w > (1-\tau)M$, where the inequality follows from $w \leq \frac12(1-\tau)M + 1$ and $M > \frac{2}{2\tau-1}$. In the second case, $d_H(1^M, 1^{M-2w}0^{2w}) = 2w > (1-\tau)M$ by the definition of $w$.
\end{proof}

\begin{definition}
\label{def:gamma-tau-m-k}
Let $\tau \in (\frac12, 1)$ and $M, k$ be positive integers. We define the constraint language $\Gamma_{\tau, M, k} = \{S^{(\tau,M,k)}\}$.
\end{definition}

\begin{lemma}
\label{lemma-general-csp-hard}
$\mathrm{CSP}(\Gamma_{\tau, M, k})$ is NP-hard whenever $k \geq 2$, $\tau \in (\frac12,1)$, and $M > \max(\frac{2}{1-\tau}, \frac{2}{2\tau-1})$.
\end{lemma}
\begin{proof}
We show that $S^{(\tau,M,k)}$ has no Siggers polymorphism, NP-hardness then follows from \cref{theorem:barto-siggers-np-hard}.

We define $p, q \in \{0, 1\}^k$ such that $p_1=0$ and $p_i=1$ for $i \geq 2$, and $q_2=0$ and $q_i=1$ for $i \neq 2$. We use $p$ and $q$ as building blocks of the following 4 elements of $S^{(\tau,M,k)}$:
$A = q^{M-2w} q^w p^w;$
$B = p^{M-2w} p^w q^w;$
$C = p^{M-2w} q^w p^w;$
$D = q^{M-2w} p^w p^w,$
illustrated in \cref{table:no-siggers}. Note that $M > \frac{2}{1-\tau}$ implies $w \leq (1-\tau)M$, so that a single disagreeing block of length $w$ is within the required distance.

\begin{table}
\centering

\begin{tabular}{r c c c c c}
$A=$ & $\overbrace{\qqtop}^{M-2w}$ & $\mid$ & $\overbrace{\qq}^{w}$ & $\mid$ & $\overbrace{\pptop}^{w}$ \\ \hline
$B=$ & $\ppmid$ & $\mid$ & $\pp$ & $\mid$ & $\qqmid$ \\ \hline
$C=$ & $\ppmid$ & $\mid$ & $\qqmid$ & $\mid$ & $\pp$ \\ \hline
$D=$ & $\qqtop$ & $\mid$ & $\pptop$ & $\mid$ & $\pptop$ \\ \hline
$f(\cdot)=$ & $z$ & $\mid$ & $z$ & $\mid$ & $z$
\end{tabular}

\caption{$k$-tuples are represented by horizontal vectors, $\vec{1}$ represents a sub-vector of $k-2$ ones. Highlighted blocks are in agreement with $1^{M-2w}0^{2w}$, ensuring that all $A, B, C, D \in S^{(\tau,M,k)}$.}
\label{table:no-siggers}
\end{table}

Let $f$ be a Siggers operation and apply it to $A,B,C,D$ element-wise. By \cref{observation:siggers-xy}, 
$$z = f(q,p,p,q) = f(q,p,q,p) = f(p,q,p,p),$$
and by \cref{observation:rep-not-rel}, $z^M \notin S^{(\tau,M,k)}$. Since 
$$\big(f(A_1,B_1,C_1,D_1), \dots, f(A_M,B_M,C_M,D_M)\big) = z^M,$$
$f$ is not a polymorphism of $S^{(\tau,M,k)}$.
\end{proof}

\begin{theorem}
\label{theorem:k-delegation-np-complete}
For any $k \geq 2$ and $\tau \in (\frac12, 1)$, $(k,\tau)$-\textsc{center} is NP-complete.
\end{theorem}
\begin{proof}
Let $M=\max(\lfloor \frac{2}{1-\tau} \rfloor, \lfloor \frac{2}{2\tau-1} \rfloor) + 1$, so by \cref{lemma-general-csp-hard}, $\mathrm{CSP}(\Gamma_{\tau, M, k})$ is NP-hard. Given an instance of $\mathrm{CSP}(\Gamma_{\tau, M, k})$ with the set of variables $U=\{v_1,\dots,v_m\}$ of size $m$ and each from the domain $\{0,1\}^k$, and the set of constraints $C=\{c_1,\dots,c_n\}$, we construct an instance of $(k,\tau)$-\textsc{center} as follows. We may assume $m \geq M$, as otherwise $m$ is bounded by a constant depending only on $\tau$ and the CSP instance can be decided in constant time.

The strings have length $m$, with one position for each variable $v_j \in U$, and there are $n$ strings, one for each constraint $c_i \in C$. Each string $x_i$, $i \in [n]$, has exactly $M$ revealed positions, so that $m_i = M$. For a constraint $c_i$ of type $S^{(\tau,M,k)}(v_{j_1},\dots,v_{j_M})$ for some distinct indices  $j_1,\dots,j_M \in [m]$, the string $x_i$ has value $1$ at the positions $j_1,\dots,j_{M-2w}$, value $0$ at the positions $j_{M-2w+1},\dots,j_M$, and $\star$ elsewhere.

Observe that there is a 1-to-1 correspondence between any assignment $f: U \rightarrow \{0,1\}^k$ of $\mathrm{CSP}(\Gamma_{\tau, M, k})$, and the centers $y^{(1)},\dots, y^{(k)}\in\{0,1\}^m$ such that $y_j^{(i)}=f(v_j)_i$, in a sense that $f$ satisfies exactly the constraints whose corresponding strings are covered by at least one of the centers $y^{(1)},\dots, y^{(k)}$. This implies that a satisfying assignment of $\mathrm{CSP}(\Gamma_{\tau, M, k})$ exists if and only if there are $k$ centers that cover all strings.
\end{proof}

Note that, as in the case of a single center, our reduction produces instances
in which every string reveals only a constant number of positions, so
$(k,\tau)$-\textsc{center} is NP-hard already for this restricted class.

We finish with the edge cases of the agreement threshold. For $\tau = 0$ the
problem is trivial, as any center covers all input strings, while for
$\tau = 1$ it is closely related to vertex colouring.

\begin{observation}
$(k, 1)$-\textsc{center} is in P for $k \leq 2$, and NP-complete for $k \geq 3$.
\end{observation}

\begin{proof}
Consider a conflict graph $G = (V, E)$ with $V=[n]$ and $E = \{(i,i') \mid \exists j \in [m] \text{ such that } x_{i,j} \neq x_{i',j} \text{ and } x_{i,j}, x_{i',j} \in \{0, 1\}\}$. The existence of a center assignment covering all input strings is equivalent to $E$ being empty for $k=1$, and to $G$ being bipartite for $k=2$.

To establish NP-hardness for $k \geq 3$, we reduce from the graph $k$-colouring problem. Given an undirected graph $G = (V, E)$, we construct an instance of $(k,1)$-\textsc{center} with $n=|V|$ strings of length $m=|E|$. For each edge $e_j = (i, i')$, we restrict the known entries at coordinate $j$ to the strings $x_i$ and $x_{i'}$, setting $x_{i,j} = 0$ and $x_{i',j} = 1$ in any order. There exists a $k$-colouring of $G$ if and only if there exist $k$ binary center strings covering all $n$ input strings.
\end{proof}

\section{Special Case: Complete Strings}

The hard instances of the previous section are sparse, as every string reveals only a constant number of positions, which raises the question of whether the missing entries are the source of hardness. Their sparsity is also what makes the reductions from CSPs work: each string became a constraint of bounded arity, as the dichotomy theorems require. A string with no missing entries instead restricts all $m$ coordinates of a center at once, so its relation has arity $m$ and grows with the instance, and the dichotomy theorems no longer apply.

Motivated by this, in this section we assume that each string $x_i$ is complete, that is, $x_i \in \{0,1\}^m$ and $m_i = m$. We refer to this special case as $(k,\tau)$-\textsc{center} \emph{with complete strings}; the radius is uniform, $\lfloor(1-\tau)m\rfloor$. We assume that $\tau \in (0,1)$ as for $\tau \in \{0, 1\}$ the problem is trivially solvable. Our results for this section are summarized in \cref{table:classification-complete}.

\begin{table}[ht]
\centering
\renewcommand{\arraystretch}{1.25}
\setlength{\tabcolsep}{5pt}

\begin{tabular}{c|c|c|c}
 & P & open & NP-hard \\ \hline
$k=1$           & $\tau<\frac12$   & $\emptyset$ & $\tau\ge\frac12$ \\
$k=2$           & $\tau\le\frac12$ & $\frac12<\tau\le\frac23$ & $\tau>\frac23$ \\
$k=3$           & $\tau\le\frac12$ & $\frac12<\tau\le\frac34$ & $\tau>\frac34$ \\
$k\in\{4,5\}$   & $\tau\le\frac12$ & $\frac12<\tau\le\frac45$ & $\tau>\frac45$ \\
$k\in\{6,7\}$   & $\tau\le\frac12$ & $\frac12<\tau\le\frac{13}{16}$ & $\tau>\frac{13}{16}$ \\
general $k$     & $\tau\le\frac12$ & $\frac12<\tau\le \tau_k^\star$ & $\tau>\tau_k^\star$
\end{tabular}
\caption{Complexity of $(k,\tau)$-\textsc{center} with complete strings
classified by $k$ and $\tau \in (0,1)$. Here $\tau_2^\star = \frac23$ and
$\tau_k^\star = \frac{5\kappa-2}{6\kappa-2}$ for $k \geq 3$;
$\kappa = \lceil \frac{k-1}{2} \rceil$. In particular, $(k,\tau)$-\textsc{center} with complete strings is NP-hard for any $\tau \in (\frac56, 1)$ and any $k$.}
\label{table:classification-complete}
\end{table}

\subsection{Single Center}

We start with $k=1$ and show that $(1,\tau)$-\textsc{center} is NP-complete
for every $\tau \in [\frac12, 1)$, even when all strings are complete. For
$\tau = \frac12$ this already follows from the proof of NP-hardness of Closest
String \citep{frances_covering_1997}, whose instances have radius exactly half
of the string length. We nevertheless state the result separately and give a
self-contained proof in the technical supplement, because we later need a
stronger property: the instances it produces have a Hamming weight barely exceeding $\frac{m}{2}$, which is one of the ideas behind our
hardness results for multiple centers.

\begin{theorem}
$(1,\frac{1}{2})$-\textsc{center} with complete strings is NP-complete, even when restricted to instances with $n$ strings of length $m = 2m'$ in which every string has Hamming weight at most $m'+4$.
\label{theorem:half-delegation-hard}
\end{theorem}

\begin{definition}
\label{def:base-voters}
Let $m \geq 6$ be an even integer. For each $i \in \{0, \dots, \frac{m}{2}\}$, we define strings $q^{(m,i)}, r^{(m,i)} \in \{0,1\}^m$:
\begin{itemize}
    \item[-] $q^{(m,i)}_j=1$ for $j \in \{2i-1, 2i\}$ and $q^{(m,i)}_j=0$ otherwise,
    \item[-] $r^{(m,i)}_j=0$ for $j \in \{2i-1, 2i\}$ and $r^{(m,i)}_j=1$ otherwise.
\end{itemize}
Moreover, let $V^{(m)} \subset \{0,1\}^m$ be the union of the following sets of binary strings
$$
\{
0^m, 1^m
\} \cup \{
q^{(m,i)}\ |\ i \in [\frac{m}{2}]
\} \cup \{
r^{(m,i)}\ |\ i \in [\frac{m}{2}]
\}.
$$
\end{definition}

\begin{lemma}
\label{lemma-d2i}
Assume that $\tau = \frac12$ and $m \geq 6$ is even. A center $y \in \{0,1\}^m$ covers all strings in $V^{(m)}$ if and only if $y_{2i-1}\neq y_{2i}$ for all $i \in [\frac{m}{2}].$
\end{lemma}

\begin{proof}

For all $i \in [\frac{m}{2}]$, the condition $y_{2i-1}\neq y_{2i}$ implies that the pair $(y_{2i-1}, y_{2i})$ is either $(0,1)$ or $(1,0)$. Every string $u \in V^{(m)}$ has the pair $(u_{2i-1},u_{2i})$ equal to either $(0,0)$ or $(1,1)$, hence they agree on exactly one of the two positions $2i-1, 2i$. Summing over all pairs, the Hamming distance between $y$ and any string in $V^{(m)}$ is exactly $\frac{m}{2}$.

For the opposite direction, assume for a contradiction that $y$ covers all strings and $y_{2i-1} = y_{2i}$ for some $i \in [\frac{m}{2}].$ Observe that $0^m, 1^m \in V^{(m)}$ together with $\tau = \frac12$, implies that $y$ has exactly $\frac{m}{2}$ zeros and $\frac{m}{2}$ ones. If $y_{2i-1} = y_{2i} = 0$, then $y$ disagrees with $r^{(m,i)}$ on positions $2i-1, 2i$ and the $\frac{m}{2}$ positions $j\in[m]$ where $y_j=1$. This implies that the Hamming distance $d_H(y, r^{(m,i)})=\frac{m}{2}+2>\tau m$ and $r^{(m,i)}$ is not covered by $y$. Otherwise, $y_{2i-1} = y_{2i} = 1$ and $q^{(m,i)}$ is not covered by $y$.
\end{proof}

Note that the $\frac{m}{2}$ strings with exactly two zeros are not necessary for the proof of \cref{lemma-d2i} but they simplify our arguments. We are now ready to show NP-hardness of $(1,\frac{1}{2})$-\textsc{center}, by a reduction from $\mathrm{CSP}(\Gamma_{\frac{1}{2}, 4})$, which is NP-hard by \cref{lemma-csp-hard}.

\begin{proof}[Proof of \cref{theorem:half-delegation-hard}]
Given an instance of $\mathrm{CSP}(\Gamma_{\frac{1}{2}, 4})$ with the set of variables $U=\{v_1,\dots,v_{m'}\}$ of size $m' \geq 3$, and the set of constraints $C=\{c_1,\dots,c_n\}$, we construct an instance of $(1,\frac{1}{2})$-\textsc{center} with complete strings as follows.

The strings have $m=2m'$ positions, two consecutive positions for each variable $v_j \in U$; and there are $n+m+2$ strings, one \emph{constraint} string for each constraint $c_i \in C$, and $m+2$ \emph{base} strings from $V^{(m)}$. Each constraint string $x_i$, $i \in [n]$, is complete and has exactly $4$ ones. If the constraint $c_i$ is of type $R^\textbf{no}(v_{j_1},v_{j_2},v_{j_3},v_{j_4})$ for some distinct indices  $j_1,j_2,j_3,j_4 \in [m]$, the string $x_i$ has value $1$ at the positions $2j_1-1,2j_2-1,2j_3-1,2j_4-1$ and value $0$ at the rest. Otherwise, the constraint $c_i$ is of type $R^\textbf{yes}(v_{j_1},v_{j_2},v_{j_3},v_{j_4})$, and $x_i$ has value $1$ at the positions $2j_1,2j_2,2j_3,2j_4$ and value $0$ at the rest. In other words, a positive literal corresponds to the pair $01$, a negative literal to its reverse, $10$; and variables not appearing in the constraint to $00$.

Consider a bijection between any assignment $f: U \rightarrow \{0,1\}$ of $\mathrm{CSP}(\Gamma_{\frac{1}{2}, m'})$, and the center $y\in\{0,1\}^m$ such that $f(v_j)=1$ if $(y_{2j-1}, y_{2j})=(0,1)$, and $f(v_j)=0$ if $(y_{2j-1}, y_{2j})=(1,0)$. \cref{lemma-d2i} implies that these are exactly the possible pairs covering the base strings.

We claim that $f$ satisfies exactly the constraints whose corresponding strings (according to the bijection described above) are covered by $y$, assuming that $y$ covers all base strings and follows the structure given by \cref{lemma-d2i}. To prove this, take the constraint string $x_i$, $i\in[n]$, corresponding to the constraint $R(v_{j_1},v_{j_2},v_{j_3},v_{j_4})$ (where $R$ is either $R^\textbf{yes}$ or $R^\textbf{no}$). If $f$ satisfies $c_i$, $y$ agrees with $x_i$ on at least 2 of the 4 pairs $(2j_1-1,2j_1), (2j_2-1,2j_2), (2j_3-1,2j_3), (2j_4-1,2j_4)$, hence $d_H(y, x_i) \leq m'$ (accounting the distance 2 for each disagreeing pair and $m'-4$ for the remaining pairs) and $x_i$ is covered. Otherwise, $f$ does not satisfy $c_i$, $y$ agrees with $x_i$ on at most 1 of the 4 pairs $(2j_1-1,2j_1), (2j_2-1,2j_2), (2j_3-1,2j_3), (2j_4-1,2j_4)$, hence $d_H(y, x_i) \geq m' + 2 > \frac{m}{2}$ and $x_i$ is not covered.

To show hardness with the weight restriction, assume $m'\ge 3$ and let $t^{(m)} = (01)^{m'}\in\{0,1\}^{m}$ be the string whose every consecutive pair equals $01$. We compute the Hamming distance of each string to $t^{(m)}$.

Every base string is at distance exactly $m'$ from $t^{(m)}$, as $01$ differs from both $00$ and $11$ in exactly one coordinate. Every constraint string has exactly 4 active pairs, one per literal, each equal to $01$ (positive literal) or $10$ (negative literal), and $m'-4$ inactive pairs equal to $00$. An inactive pair $00$ differs from $01$ in one coordinate; an active $01$ differs
in $0$; an active $10$ differs in $2$ coordinates. Hence the distance of every constraint string to $t^{(m)}$ is in the interval $[m'-4, m'+4]$.

Now translate the whole instance by $t^{(m)}$: replace each string $x$ by
$u = x \oplus t^{(m)}$, and note that for any center $y$ we have
$d_H\!\big(y, x\big) = d_H\!\big(y\oplus t^{(m)},\, u\big)$, preserving the feasibility of the instance. In the translated instance, the weight of each string equals its former distance to $t^{(m)}$: base strings now have weight exactly $m'$ and constraint strings' weight is between $m'-4$ and $m'+4$. In particular, every string has Hamming weight at most $m'+4$.
\end{proof}

For $\tau > \frac12$, we get NP-hardness by a simple reduction from $(1,\frac{1}{2})$-\textsc{center} where we append the input strings by zeros.

\begin{corollary}
\label{corollary:1-del-hard-complete}
$(1,\tau)$-\textsc{center} with complete strings is NP-complete for any $\tau \in [\frac12, 1)$.
\end{corollary}

\begin{proof}
Let $\tau > \frac12$ be fixed. Given an instance $\mathcal{I}$ of $(1,\frac{1}{2})$-\textsc{center} with $n$ complete strings of length $m$, we construct an instance $\hat{\mathcal{I}}$ of $(1,\tau)$-\textsc{center} with $n$ complete strings of length $\hat{m}$, where $\hat{m} = \lfloor \frac{m}{2(1-\tau)} \rfloor$.

The first $m$ positions of each string are preserved, and the additional $\hat{m} - m$ positions are all zero. That is, given the original string $x^{(i)} \in \{0,1\}^m$, $i \in [n]$, the corresponding string $\hat{x}^{(i)} \in \{0,1\}^{\hat{m}}$ in $\hat{\mathcal{I}}$ is defined as $\hat{x}^{(i)}_j = x^{(i)}_j$ for $j \leq m$ and $\hat{x}^{(i)}_j = 0$ for $j > m.$

If there is a center $\hat{y}$ covering all strings of length $\hat{m}$, consider the center $y=(\hat{y}_1,\dots,\hat{y}_m)$ obtained by projecting $\hat{y}$ onto the first $m$ coordinates. Any string agrees with $\hat{y}$ on at least $\lceil \tau \hat{m} \rceil$ positions, hence they agree on at least $\lceil \tau \hat{m} \rceil - (\hat{m}-m)$ out of the first $m$ positions. Because $\lceil \tau \hat{m} \rceil - (\hat{m}-m) = m - \lfloor (1-\tau) \hat{m} \rfloor \geq \frac{m}{2}$, where the inequality follows from the definition of $\hat{m}$. Hence, $y$ covers all strings in $\mathcal{I}$.

Now assume that it is not possible to cover all strings in $\hat{\mathcal{I}}$ with a single center. To prove that the same holds for $\mathcal{I}$, let $y$ be a center covering all strings in $\mathcal{I}$ and let $\hat{y}$ be the center obtained from $y$ by appending $\hat{m} - m$ zeros to $y$. Since for any string $\hat{x}^{(i)}$, $y$ agrees with $x^{(i)}$ on at least $\lceil \frac{m}{2} \rceil$ positions, $\hat{y}$ agrees with $\hat{x}^{(i)}$ on at least $\lceil \frac{m}{2} \rceil +(\hat{m}-m) = \hat{m} - \lfloor \frac{m}{2} \rfloor \geq \tau \hat{m}$ positions, where the inequality follows from the definition of $\hat{m}$. Hence, $\hat{y}$ covers all strings in $\hat{\mathcal{I}}$.
\end{proof}

Next, we consider $\tau < \frac12$ and show that \textsc{max}-$(1,\tau)$-\textsc{center} with complete strings can be solved exactly in polynomial time. The argument is probabilistic: a center selected uniformly at random covers any fixed string with probability exponentially close to one, so as long as the strings are long enough, the expected number of uncovered strings is smaller than one, and a center covering all of them can be found by the method of conditional expectations; short instances are solved by exhaustive search. The result also settles $(1,\tau)$-\textsc{center} with complete strings, and carries over to Close to Most Strings.

\begin{lemma}
    Let $\tau = \frac12-\epsilon\;$ for some $\epsilon \in (0, \frac12)$ and denote $\gamma = \min\{\epsilon^2, 1-H(\frac14) \approx 0.18\}$. For any instance of \textsc{max}-$(1,\tau)$-\textsc{center} with $n$ strings of length $m$ satisfying $m_i \geq \frac{\log_2 n}{\gamma}$ for all $i \in [n],$ there exists a center covering all strings and we can compute it in polynomial time.
\label{lemma:cover-all}
\end{lemma}

\begin{observation}
\label{observation:random-drep}
    Let $\tau \in (0, \frac12),$ and let $p_i$ be the probability that a center selected uniformly at random covers the string $x_i$. Then,
    $$
    p_i > 1 - 2^{m_i(H(\tau) - 1)}.
    $$
\end{observation}
\begin{proof}
    \begin{align*}
        p_i
        &= \mathrm{Pr}[\text{$y$ covers $x_i$}] \\
        &= 1 - \sum_{j=0}^{\lceil \tau m_i\rceil - 1}{\frac{\binom{m_i}{j}}{2^{m_i}}}  \\
        &\geq 1 - \frac{1}{2^{m_i}} 2^{m_iH(\frac{\lceil \tau m_i\rceil - 1}{m_i})} \\
        &> 1 - \frac{1}{2^{m_i}}2^{m_iH(\tau)} \\
        &= 1 - 2^{m_i(H(\tau) - 1)}.
    \end{align*}
    The first inequality follows from the standard binary entropy bound on the sum of the first binomial coefficients, and the second inequality follows from $\lceil x \rceil < x + 1$ for any $x$.
\end{proof}

\begin{observation}
\label{observation:derandomized-drep}
For any instance of \textsc{max}-$(1,\tau)$-\textsc{center} with $n$ strings of
length $m$, let $X^+$ be the random variable denoting the number of strings
covered by a center selected uniformly at random. Then there is a center $y$
covering at least $\lceil \mathbb{E}[X^+] \rceil$ strings, and a deterministic
polynomial-time algorithm that computes $y$.
\end{observation}
\begin{proof}
    We can find the center $y$ using the method of conditional expectations. In short, we construct $y$ in $m$ steps, where in the $j$-th step, we select the $j$-th element $y_j$ that maximizes the expected value of $X^+$ conditional on the first $j$ elements being fixed. This can be done in polynomial time as
\begin{align*}
    \mathbb{E}[X^+ | y_{(1,\dots,j)}]
    &= \sum_{i \in [n]}{\mathrm{Pr}[\text{$y$ covers $x_i$} | y_{(1,\dots,j)}]} \\
    &= n - \sum_{i \in [n]}\frac{1}{2^{m_{ij}}}\sum_{k=0}^{\Delta(i)}{\binom{m_{ij}}{k}},
\end{align*}
where $m_{ij}$ is the number of revealed positions of $x_i$ among $\{j+1,\dots,m\}$, that is, the revealed positions not yet fixed, and $\Delta(i) = \lceil \tau m_i \rceil - 1 - \sum_{k=1}^{j}{[y_k=x_{ik}]}$ is the maximum number of positions from $\{j+1, \dots, m\}$ on which $y$ can agree with $x_i$ while still not covering $x_i$; for $y$ already covering $x_i$, $\Delta(i) < 0$ and the inner sum is empty, $\sum_{k=0}^{\Delta(i)}{\binom{m_{ij}}{k}} = 0$.
\end{proof}

\begin{proof}[Proof of \cref{lemma:cover-all}]
    Let $y \in \{0,1\}^m$ be a center selected uniformly at random, let $X^-=n-X^+$ be the random variable denoting the number of strings not covered by $y$, and let $M=\min_{i \in [n]}{m_i} \geq \frac{\log_2 n}{\gamma}$. We show that $\mathbb{E}[X^-] < 1$ below, the statement then follows from \cref{observation:derandomized-drep}.
    \begin{align}
        \mathbb{E}[X^-]
        \label{eq:all-1} &= \sum_{i \in [n]}{(1 - \mathrm{Pr}[\text{$y$ covers $x_i$}])} \\
        \label{eq:all-2}&< \sum_{i \in [n]}{2^{m_i(H(\tau)-1)}} \\
        \label{eq:all-3}&\leq n 2^{M(H(\tau)-1)} \\
        \label{eq:all-6}&\leq \frac{n}{2^{M\gamma}} \\
        \label{eq:all-7} &\leq 1.
    \end{align}
    Here, \eqref{eq:all-1} follows from linearity of expectation, \eqref{eq:all-2} follows from \cref{observation:random-drep}, and \eqref{eq:all-6} follows from the standard bound on the binary entropy function $H(\frac12-x) \leq 1-x^2$ for any $ x \in [0, \frac14]$ and $H(\frac12-x) < H(\frac14)$ for any $x \in (\frac14, \frac12).$
\end{proof}

\begin{corollary}
    For any $\tau \in (0, \frac12)$, there is a polynomial-time algorithm for \textsc{max}-$(1,\tau)$-\textsc{center} with complete strings.
\label{corollary:optimal-complete}
\end{corollary}
\begin{proof}
    If $m < \frac{\log_2 n}{\gamma},$ a center covering all strings may not exist, but we can find the best center by trying all the $2^m < n^{1/\gamma}$ candidates. Otherwise, $m \geq \frac{\log_2 n}{\gamma},$ and we apply \cref{lemma:cover-all} to obtain an algorithm that computes the center covering all strings in polynomial time.
\end{proof}

\subsection{Multiple Centers}

Next, we aim to classify $(k,\tau)$-\textsc{center} with complete strings for $k \geq 2$. We reduce from the weight-restricted $(1,\frac12)$-\textsc{center} of \cref{theorem:half-delegation-hard} by adding $k-1$ outlier
strings that are far from every original string, and also far from each other. The latter property requires the outliers to form a code of large minimum distance, which we obtain from the standard construction in the following lemma. This construction is tight, in the sense that no code with $k$ codewords attains a larger ratio of minimum distance to its length.

\begin{lemma}
\label{lemma-code-k}
For any $k \geq 2,$  there is a code $\mathcal{C}_k=\{c_1, \dots, c_k\} \subset \{0,1\}^{m_k}$ of length $m_k=\binom{k}{\lfloor \frac{k}{2} \rfloor}$ and minimum distance $d_k=2\binom{k-2}{\lfloor \frac{k}{2} \rfloor - 1}=\frac{\lceil \frac{k}{2} \rceil}{2\lceil \frac{k}{2} \rceil - 1}m_k$, where all codewords have a Hamming weight $w_k \geq \frac{m_k}{2}$.
\end{lemma}

\begin{proof}
Let $S_k=\{s_1, \dots, s_{m_k}\} \subset \{0, 1\}^k$ be the set of all binary
strings of length $k$ with exactly $\lceil \frac{k}{2} \rceil$ ones. Let $\mathbf{S}$ be the matrix whose columns are the elements of $S_k$, and let
$\mathcal{C}_k$ consist of the rows of $\mathbf{S}$, so that $c_{i,j} = s_{j,i}$ for all $i \in [k]$ and $j \in [m_k]$, see \cref{table:code-C4} for an example. For any two distinct codewords $a \neq b \in \mathcal{C}_k,$ their Hamming distance
$$d_H(a, b) = |\{i \in [m_k]: a_i > b_i\}| + |\{i \in [m_k]: a_i < b_i\}|,$$
which is equal to $d_k = 2\binom{k-2}{\lfloor \frac{k}{2} \rfloor - 1}$.

\begin{table}
\centering
\begin{tabular}{c|cccccc}
      & $s_1$ & $s_2$ & $s_3$ & $s_4$ & $s_5$ & $s_6$ \\\hline
$c_1$ & 1 & 1 & 1 & 0 & 0 & 0 \\
$c_2$ & 1 & 0 & 0 & 1 & 1 & 0 \\
$c_3$ & 0 & 1 & 0 & 1 & 0 & 1 \\
$c_4$ & 0 & 0 & 1 & 0 & 1 & 1 \\
\end{tabular}
\caption{Example of the construction of $\mathcal{C}_4$.}
\label{table:code-C4}
\end{table}

Finally, all codewords of $\mathcal{C}_k$ have the same Hamming weight by symmetry, and since $\mathbf{S}$ contains $m_k \lceil \frac{k}{2} \rceil$ ones in total, this weight is $w_k = \frac{m_k}{k}\lceil \frac{k}{2} \rceil \geq \frac{m_k}{2}$.
\end{proof}

By the Plotkin bound (\cref{theorem:plotkin}), we obtain $\frac{d_k}{m_k} \leq \frac12 \frac{k}{k-1},$ leaving a small gap between the ratio $\frac{d_k}{m_k} = \frac{\lceil \frac{k}{2} \rceil}{2\lceil \frac{k}{2} \rceil - 1}$. However, a more careful proof of the Plotkin bound shows that our construction from \cref{lemma-code-k} is tight.

\begin{lemma}
\label{lemma:plotkin-tight}
For any code $\mathcal{C} = \{c_1, \dots, c_k\} \subset \{0, 1\}^m$ of length $m$ and minimum distance $d$, $\frac{d}{m} \leq \frac{\lceil \frac{k}{2} \rceil}{2\lceil \frac{k}{2} \rceil - 1}.$
\end{lemma}
\begin{proof}
As in the proof of Plotkin bound, we double count the sum of Hamming distances between codewords, that is, the value $S = \sum_{1 \leq i < j \leq k}{d_H(c_i, c_j)}$. Summing by coordinates, 
$$S = \sum_{i \in [m]}{|\{j \in [k]: c_{j, i}=0\}| |\{j \in [k]: c_{j, i}=1\}|} \leq m \lfloor \frac{k}{2} \rfloor \lceil \frac{k}{2} \rceil.$$ Because $d_H(c_i, c_j) \geq d,$ we have $S \geq \binom{k}{2}d.$ Put together, $\binom{k}{2}d \leq m \lfloor \frac{k}{2} \rfloor \lceil \frac{k}{2} \rceil$, hence
$$\frac{d}{m} \leq \frac{\lfloor \frac{k}{2} \rfloor \lceil \frac{k}{2} \rceil}{\binom{k}{2}} = \frac{\lceil \frac{k}{2} \rceil}{2\lceil \frac{k}{2} \rceil - 1}.$$
\end{proof}

\begin{theorem}
\label{thm:code-outliers}
Let $k \geq 2$ and let $\kappa = \lceil \frac{k-1}{2} \rceil$.
Then $(k,\tau)$-\textsc{center} with complete strings is NP-hard for any
$\tau_k^\star < \tau < 1$, where $\tau_2^\star = \frac23$ and
$\tau_k^\star = \frac{5\kappa-2}{6\kappa-2}$ for $k \geq 3.$
\end{theorem}

\begin{proof}
Assume first that $k \geq 3$; $k=2$ is treated at the end of the proof.

We reduce from the weight-restricted $(1,\tfrac12)$-\textsc{center} with
complete strings of \cref{theorem:half-delegation-hard}. Given an
instance with $n$ strings $x_1,\dots,x_n \in \{0,1\}^{2m'}$ of length $2m'$,
each of Hamming weight at most $m'+4$, we construct an instance of $(k,\tau)$-\textsc{center} with $n + k - 1$ strings of length $m = \big\lceil m'/(1-\tau) \big\rceil$: the \emph{real} strings
$\tilde{x}_i = x_i \,|\, 0^{m-2m'}$ for $i \in [n]$, and the \emph{outlier}
strings $z_1,\dots,z_{k-1}$ defined further below.  The length $m$ is chosen so that  the maximum distance at threshold $\tau$ is $$d_\text{max} = \lfloor (1-\tau)m \rfloor = m'.$$ 

Since $\tau > \tau_k^\star$ is equivalent to $\frac{1}{1-\tau} > 6 - \frac{2}{\kappa}$, the value
$\delta = \frac{1}{1-\tau} - (6 - \frac{2}{\kappa})$ is a positive constant. Denote $m_{k-1} = \binom{k-1}{\lfloor (k-1)/2 \rfloor}$ and assume that $m' \geq (m_{k-1} + 5)/\delta$, as otherwise we can decide the instance by brute force.

Consider the code $\mathcal{C}_{k-1} = \{c_1,\dots,c_{k-1}\} \subset \{0,1\}^{m_{k-1}}$ of \cref{lemma-code-k}, with minimum distance $d_{k-1} = 2\binom{k-3}{\lfloor (k-1)/2 \rfloor - 1}$ satisfying
$\frac{d_{k-1}}{m_{k-1}} = \frac{\kappa}{2\kappa-1}$, and every codeword having Hamming weight 
$w_{k-1} \geq \frac{m_{k-1}}{2} \geq \frac{d_{k-1}}{2}$.

For each $t \in [k-1],$ we define the outlier
$$
z_t = 1^{\,m - r m_{k-1}} \,\big|\, (c_t)^r,
$$
in other words, $z_t$ contains all ones followed by $r$ copies of the
codeword $c_t$ on the last $r m_{k-1}$ positions. Here, $r = \big\lceil (2m'+1)/d_{k-1} \big\rceil$ is chosen so that no single center can cover two outliers simultaneously, as for any two outliers $z_t \neq z_{t'}$,
$$
d_H(z_t, z_{t'}) = r \, d_H(c_t, c_{t'})
\geq r d_{k-1} \geq 2m'+1 > 2 d_\text{max}.
$$

We now verify that the first part of each outlier contains enough ones,
namely that $m - r m_{k-1} \geq 2m'+5$, or equivalently, that
$m - 2m' \geq r m_{k-1} + 5$. Indeed,
\begin{align}
r m_{k-1} + 5
\label{eq:tail-1} &\leq \Big(\frac{2m'}{d_{k-1}} + 1\Big)m_{k-1} + 5 \\
\label{eq:tail-2} &= \Big(4 - \frac{2}{\kappa}\Big)m' + m_{k-1} + 5 \\
\label{eq:tail-3} &\leq \Big(4 - \frac{2}{\kappa}\Big)m' + \delta m' \\
\label{eq:tail-4} &= \frac{m'}{1-\tau} - 2m' \\
\label{eq:tail-5} &\leq m - 2m'.
\end{align}
Here, \eqref{eq:tail-1} follows from 
$r = \big\lceil \frac{2m'+1}{d_{k-1}} \big\rceil \leq \frac{2m'+d_{k-1}}{d_{k-1}}$,
\eqref{eq:tail-2} from $\frac{m_{k-1}}{d_{k-1}} = \frac{2\kappa-1}{\kappa}$,
and \eqref{eq:tail-3} from the assumption $m' \geq (m_{k-1}+5)/\delta$. Finally, \eqref{eq:tail-4} is from the definition of $\delta$ and \eqref{eq:tail-5} from the definition of $m$.

Observe that the weight of $z_t$ is
\begin{align}
m - r m_{k-1} + r w_{k-1}
\label{eq:wz-1} &\geq 2m' + 5 + \frac{r d_{k-1}}{2} \\
\label{eq:wz-2} &\geq 2m' + 5 + \frac{2m'+1}{2} \\
\label{eq:wz-3} &> 3m' + 5
\end{align}

Here, \eqref{eq:wz-1} follows from \eqref{eq:tail-5} and $w_{k-1} \geq \frac{d_{k-1}}{2}$, and \eqref{eq:wz-2} follows from the definition of $r$.
Combined with the fact that the weight of each real string $\tilde{x}_i$ is bounded by $m' + 4,$ we obtain $d_H(z_t, \tilde{x}_i) \geq 2m'+1,$ implying that no single center can cover both an outlier $z_t$ and a real string $\tilde{x}_i$.

Let $y \in \{0,1\}^{2m'}$ be a center covering all strings in the original instance. Then $y \,|\, 0^{m-2m'}$ covers all real strings, and the $k-1$ centers $z_1,\dots,z_{k-1}$ cover the outliers.

For the opposite direction, assume that $k$ centers cover all strings in the reduced
instance. Every outlier is covered by some center, no center covers two outliers, and no center covers both an outlier and a real string. This means that $k-1$ distinct centers are consumed by the outliers, and the single remaining center covers all real strings, and its restriction to the first $2m'$ positions is a
center covering all strings in the original instance.

For the case $k=2$, the idea is to use a single outlier consisting of all ones.
Since $\tau > \tau_2^\star = \frac23$
is equivalent to $\frac{1}{1-\tau} > 3$, the value
$\delta = \frac{1}{1-\tau} - 3$ is a positive constant; assume
$m' \geq 5/\delta$, as otherwise we can brute force. Set
$m = \big\lceil m'/(1-\tau) \big\rceil$, so that
$d_\text{max} = \lfloor (1-\tau)m \rfloor = m'$ exactly as before, and let the
reduced instance consist of the real strings
$\tilde{x}_i = x_i \,|\, 0^{m-2m'}$ for $i \in [n]$ and the single outlier
$z = 1^m$. The weight of $z$ is
$$
m \geq \frac{m'}{1-\tau} = 3m' + \delta m' \geq m' + 5,
$$
so $d_H(z, \tilde{x}_i) \geq (3m'+5) - (m'+4) = 2m'+1 > 2d_\text{max}$ for every real
string $\tilde{x}_i$, hence $z$ needs its own center.
\end{proof}

\section{Discussion and Future Work}

The main question left open from our results is the (in)tractability of $(k,\tau)$-\textsc{center} with complete strings in the open region $\tau \in (\frac12, \tau_k^\star]$ for $k \geq 2$; resolving the complexity of the problem in this region appears to be technically quite challenging. Furthermore, our impossibility results for the problem motivate the study of \emph{approximations} to the \textsc{max}-$(k,\tau)$-\textsc{center} variant of the problem, starting from the simplest case of $k=1$ without missing entries. To this end, it is fairly easy to see that for $\tau \leq 1/2$, either any center or its complement will attract any voter, and hence the best of the two options will be a $2$-approximation. Interestingly, any improvement over this bound, or any constant approximation for $\tau > 1/2$ (even $\tau = 1/2 + \varepsilon$ for any $\varepsilon > 0$) seems to be elusive. As we show in the Appendix, a constant approximation is achievable for  ``sparse'' instances, i.e., when the number of non-missing entries per string is a constant. This is achieved via leveraging a similar dichotomy theorem to those we used for our main NP-completeness results, due to \citet{khanna_approximability_2001}. The same dichotomy also establishes the APX-hardness of \textsc{max}-$(k,\tau)$-\textsc{center}. 

Another interesting question regards parameterized results for the problem. While these were not the focus of our main results, in the Appendix we remark that when either $m$ or $n$ is constant, the problem is solvable in polynomial time; in the former case this is via a straightforward enumeration, whereas in the latter case this is achieved via either an appropriate ILP formulation or via a dynamic programming approach. In terms of more nuanced parameters, our NP-hardness results shown in \cref{table:general-classification} hold even when the number of missing entries of each string is a small constant. Imposing further restrictions on the structure of the input strings and obtaining positive results is an interesting direction for future work; to this end, in the Appendix we show how an FPT algorithm in an appropriate parameter called the \emph{vertex cover number}, due to \cite{friedrich_binary_2025} for the binary $k$-center problem, can be adapted for our problem as well.  

Finally, we remark that investigating the computational complexity of \textsc{max}-$(k,\tau)$-\textsc{center} roughly amounts to answering the following question: ``For any set of centers that has been constructed in polynomial time, how many strings can this set cover, compared to the maximum possible number of voters that any set can cover?'' A perhaps even more fundamental question is of an information-theoretic nature: ``Given a set of binary strings, how many of them can any set of centers \emph{feasibly} cover?''. In the worst case, it is not difficult to see that for large enough thresholds, any constant-size such set can only cover a constant number of strings, see the Appendix for a proof. An interesting investigation would be to pose this question in a setting where the input strings are not chosen worst-case, but are drawn from some distribution.

\section*{Acknowledgments}
Jakub Dargaj was supported by the UK Engineering and Physical Sciences Research Council (EPSRC) DTA Scholarship EP/W524384/1.
Aris Filos-Ratsikas was supported by the EPSRC grant EP/Y003624/1.
Paul W. Goldberg was supported by the EPSRC grant EP/X038548/1.

\bibliographystyle{plainnat}
\bibliography{references}

\clearpage
\appendix

\section{Instances with Few Coverable Strings}

For any $\tau > 1/2$, we describe a class instances of \textsc{max}-$(k,\tau)$-\textsc{center} with the number of strings exponential in the length, in which the number of voters covered by any center is bounded by a constant that depends on $\tau$ but is independent of the input size. Our tools are the Gilbert-Varshamov bound (\cref{theorem:g-v}) and the Johnson bound (\cref{theorem:johnson}) from coding theory.

\begin{lemma}
\label{lemma:no-cover-1/2}
Let $\tau \in (\frac12, 1)$ and denote $\tau = \frac{1}{2}(1+\epsilon)$ for some
$\epsilon \in (0, 1)$. For any $m \geq \frac{8}{\epsilon^2}$, there is an instance
of \textsc{max}-$(1,\tau)$-\textsc{center} with complete strings, consisting of
$n = 2^{\Theta(m)}$ strings of length $m$, such that no center covers more
than $\frac{2}{\epsilon^2}$ strings.
\end{lemma}

\begin{proof}
The input strings of our instance are the codewords of a binary code
$\mathcal{C} \subseteq \{0,1\}^m$ of maximum size with minimum distance
$d_\mathcal{C} = \frac{1}{2}(1-\delta)m$ for some $0 < \delta < 1$ fixed later.

Let $\mu = \lceil \tau m \rceil$ be the minimum agreement, and let
$d_\mathrm{max} = m - \mu = \frac{1}{2}(1-\gamma)m$, for some
$0 < \gamma < 1$, be the maximum distance between a center and a string it
covers.

First, we bound the value of $\gamma$ using $\epsilon$. As
$\tau m \leq \mu < \tau m + 1$ and the assumption $m \geq \frac{8}{\epsilon^2}$
implies $m > \frac{2}{\epsilon}$, we obtain
$$
\frac{1}{2}(1-\epsilon) \geq \frac{d_\mathrm{max}}{m} > \frac{1}{2}(1-2\epsilon),
$$
together with $d_\mathrm{max} = \frac{1}{2}(1-\gamma)m$ implying
\begin{align}
\epsilon \leq \gamma < 2\epsilon.
\label{bound:gamma-epsilon}
\end{align}
Second, we set $d_\mathcal{C} = \lfloor \frac{1}{2}(1-\frac{1}{4}\epsilon^2)m
\rfloor$ and bound the value of $\delta$ using $\epsilon$. The floor function
implies $\frac{1}{2}(1-\frac{1}{4}\epsilon^2)m - 1 < \frac{1}{2}(1-\delta)m
\leq \frac{1}{2}(1-\frac{1}{4}\epsilon^2)m$. From the second inequality we get
$\delta \geq \frac{1}{4}\epsilon^2$. The first inequality is equivalent to
$\delta < \frac{1}{4}\epsilon^2 + \frac{2}{m}$, which together with our
assumption $m \geq \frac{8}{\epsilon^2}$ implies
$\delta < \frac{1}{2}\epsilon^2$. Altogether,
\begin{align}
\frac{1}{4}\epsilon^2 \leq \delta < \frac{1}{2}\epsilon^2.
\label{bound:delta-epsilon}
\end{align}
The bounds on $\gamma$ and $\delta$ give $\gamma^2 \geq \epsilon^2 > \delta$,
so we can apply the Johnson bound (\cref{theorem:johnson}) to obtain
$$
|B(y,d_\mathrm{max}) \cap \mathcal{C}| \leq
\min\Big(m, \frac{1-\delta}{\gamma^2-\delta}\Big)
$$
for any center $y \in \{0,1\}^m$. Plugging the inequalities
\eqref{bound:gamma-epsilon} and \eqref{bound:delta-epsilon} into the
expression above we obtain
$$
\frac{1-\delta}{\gamma^2-\delta}
< \frac{1-\frac{1}{4}\epsilon^2}{\epsilon^2-\frac{1}{2}\epsilon^2}
< \frac{2}{\epsilon^2}.
$$
Observe that $|B(y,d_\mathrm{max}) \cap \mathcal{C}|$ is precisely the number
of input strings covered by $y$, hence no center covers more than
$\min(m, \frac{2}{\epsilon^2})$ of them.

It remains to bound the number of input strings. Denote
$\alpha = \frac{d_\mathcal{C}-1}{m}$ and apply the Gilbert-Varshamov bound
(\cref{theorem:g-v}) to obtain
$|\mathcal{C}| \geq 2^{m(1-H(\alpha))}$. Finally,
$$
\alpha = \frac{d_\mathcal{C}-1}{m} < \frac{1}{2}(1-\delta)
\leq \frac{1}{2} - \frac{1}{8}\epsilon^2,
$$
implying $n = |\mathcal{C}| \geq 2^{m(1-H(1/2-\frac{1}{8}\epsilon^2))}
= 2^{\Theta(m)}$.
\end{proof}

For $\tau > 3/4$, we can strengthen our bounds so that no center covers more
than a single input string.

\begin{lemma}
\label{lemma:no-cover-3/4}
Let $\tau \in (\frac34, 1)$ and denote $\tau = \frac{1}{2}(\frac{3}{2}+\epsilon)$ for
some $\epsilon \in (0, \frac12)$. For any $m \geq \frac{2}{\epsilon}$, there is an
instance of \textsc{max}-$(1,\tau)$-\textsc{center} with complete strings, consisting
of $n = 2^{\Theta(m)}$ strings of length $m$, such that no center covers more
than one input string.
\end{lemma}

\begin{proof}
Let $\mu = \lceil \tau m \rceil$ be the minimum agreement, and
$d_\mathrm{max} = m - \mu$ the maximum distance between a center and a string
it covers. The input strings are the codewords of a binary code
$\mathcal{C} \subseteq \{0,1\}^m$ of maximum size with minimum distance
$d_\mathcal{C} = 2d_\mathrm{max}+1$; then
$|B(y,d_\mathrm{max}) \cap \mathcal{C}| \leq 1$ for every center
$y \in \{0,1\}^m$ by the triangle inequality, that is, no center covers more
than one input string.

It remains to bound $n = |\mathcal{C}|$. Denote
$\alpha = \frac{d_\mathcal{C}-1}{m}$ and observe that
\begin{align*}
\alpha 
&= \frac{2d_\mathrm{max}}{m} = \frac{2(m-\mu)}{m} \\
&\leq \frac{2(m-\tau m)}{m} = 2(1-\tau) = \frac{1}{2} - \epsilon,
\end{align*}

using $\mu \geq \tau m$ and $\tau = \frac34 + \frac{\epsilon}{2}$. In
particular $d_\mathcal{C} \leq (\frac12 - \epsilon)m + 1 \leq \frac{m}{2}$
because $m \geq \frac{2}{\epsilon}$, so the Gilbert-Varshamov bound
(\cref{theorem:g-v}) applies and yields
$|\mathcal{C}| \geq 2^{m(1-H(\alpha))} \geq
2^{m(1-H(1/2-\epsilon))} = 2^{\Theta(m)}$.
\end{proof}

We have identified instances in which no single center covers more than a
constant number of input strings, for any $\tau > 1/2$. This extends
immediately to any number of centers.

\begin{theorem}
\label{theorem:no-cover-k}
Let $\tau \in (\frac12, 1)$ and denote $\epsilon = 2\tau - 1$. For any integer
$k \geq 1$ and any $m \geq \frac{8}{\epsilon^2}$, there is an instance of
\textsc{max}-$(k,\tau)$-\textsc{center} with $n = 2^{\Theta(m)}$ complete strings of length $m$, in which no $k$ centers cover more than $\frac{2k}{\epsilon^2} = O(k)$ strings. If moreover $\tau > 3/4$ and $m \geq \frac{4}{2\epsilon-1}$, there is such an instance in which no $k$
centers cover more than $k$ strings.
\end{theorem}

\begin{proof}
Let $\mathcal{C}$ be the instance guaranteed by
\cref{lemma:no-cover-1/2}, and denote
$d_\mathrm{max} = m - \lceil \tau m \rceil$. Let $Y$ be any set of $k$
centers and let
$B_Y = \{x \in \mathcal{C} : \min_{y \in Y} d_H(x,y) \leq d_\mathrm{max}\}$
be the set of input strings they cover. As
$B_Y = \bigcup_{y \in Y} \big(B(y,d_\mathrm{max}) \cap \mathcal{C}\big)$ and
$|B(y,d_\mathrm{max}) \cap \mathcal{C}| \leq \frac{2}{\epsilon^2}$ for every
$y \in Y$ by \cref{lemma:no-cover-1/2}, we obtain
$$
|B_Y| \leq \sum_{y \in Y}|B(y,d_\mathrm{max}) \cap \mathcal{C}|
\leq k \frac{2}{\epsilon^2} = O(k).
$$
For the second statement, note that $\tau > 3/4$ is equivalent to
$\epsilon > \frac12$, and that $\tau = \frac12(\frac32 + \epsilon')$ for
$\epsilon' = \frac{2\epsilon-1}{2} \in (0,\frac12)$. Since
$m \geq \frac{4}{2\epsilon-1} = \frac{2}{\epsilon'}$, we may take the instance
of \cref{lemma:no-cover-3/4} instead, which satisfies
$|B(y,d_\mathrm{max}) \cap \mathcal{C}| \leq 1$ for every center $y$; the same
computation then gives $|B_Y| \leq k$.
\end{proof}

\section{Small Instances}

\begin{theorem}
\label{theorem:small-inputs}
    $(k,\tau)$-\textsc{center} admits an algorithm with runtime $n^{\mathcal{O}(n^2)} \mathcal{O}(\log m).$
\end{theorem}
\begin{proof}
    To solve the $(k, \tau)$-\textsc{center}, we iterate over all possible partitions of the $n$ input strings into $k$ disjoint clusters; the number of such partitions is bounded by $k^n$. For a fixed partition $P=\{V_1, \dots, V_k\}$, the center string $y$ is responsible for covering the subset of strings in $V_y$ only. This allows us to split the general $(k,\tau)$-\textsc{center} into $k$ independent $(1,\tau)$-\textsc{center} instances.

    Although $(k,\tau)$-\textsc{center} is more general than Hamming radius clustering, we can easily phrase $(1,\tau)$-\textsc{center} as a $\delta$-\textsc{multi strings} problem by \citet{knop_combinatorial_2020}, with per-string Hamming distance lower bounded by $0$ and upper-bounded by $\lfloor (1-\tau)m_i\rfloor$. This immediately shows the existence of an algorithm for $(1,\tau)$-\textsc{center} with running time $n^{\mathcal{O}(n^2)} \mathcal{O}(\log m).$ The instance of $(k, \tau)$-\textsc{center} is a YES-instance if and only if there is a partition $P=\{V_1,\dots,V_k\}$ such that the output of $\delta$-multi strings is positive for all $k$ clusters of P.
\end{proof}

Although \cref{theorem:small-inputs} solves the case with constant number of input strings, we also give an elementary dynamic program for it. Its dependence on $m$ is worse, but the algorithm is self-contained, avoiding the integer programming machinery behind $\delta$-\textsc{multi strings}, and it extends directly to the maximization version.

\begin{theorem}
    \textsc{max}-$(k,\tau)$-\textsc{center} admits an algorithm with runtime $m^{\mathcal{O}(n)}$.
\end{theorem}
\begin{proof}
See \cref{algo:k-tau-conn}. For each $j \in \{0, \dots, m\},$ we maintain a set $\dpmap{j}$ containing distance matrices $\mathbf{D} \in \{0, \dots, m\}^{n \times k}$, where $\dpmap{j}$ considers the first $j$ positions of the input. 

The distance matrix
$
\mathbf{D} \in \dpmap{j}
$
if there are $k$ centers $y_1, \dots, y_k \in \{0, 1\}^j$ such that $\mathbf{D}_{i, l} = d_H(x_{i, :j}, y_l)$ is the Hamming distance of center $y_l$ from string $x_i$ on the first $j$ positions, for all $i \in [n]$ and $l \in [k]$.

The algorithm can be easily adapted to \textsc{max}-$(k,\tau)$-\textsc{center}. Instead of testing whether some matrix $\mathbf{D} \in \dpmap{m}$ covers all input strings, we compute for each such matrix the number of strings $x_i$ with $\min_{l \in [k]} \mathbf{D}_{i,l} \leq \lfloor(1-\tau)m_i\rfloor$ and return the maximum over $\dpmap{m}$. To output the centers themselves, we store with every matrix of $\dpmap{j}$ the bit vector $\mathbf{\Delta}$ that produced it, together with a pointer to its predecessor in $\dpmap{j-1}$; backtracking from the best matrix of $\dpmap{m}$ then recovers the $k$ centers position by position. Both modifications add $\mathcal{O}(1)$ work per stored matrix and leave the runtime unchanged.
\end{proof}

\begin{algorithm}[h]
\caption{Dynamic program for $(k,\tau)$-\textsc{center}}
\label{algo:k-tau-conn}
\begin{algorithmic}[1]
\REQUIRE $x_1, \dots, x_n \in \{0,1,\star\}^m$
\ENSURE \texttt{TRUE} iff $k$ centers covering $x_1, \dots, x_n$ exist
\STATE $\dpmap{0} \leftarrow \{0^{n \times k}\}$
\FORALL{$j \in [m]$}
  \STATE $\dpmap{j} \leftarrow \emptyset$
  \FORALL{$\mathbf{D} \in \dpmap{j-1}$}
    \FOR{$\mathbf{\Delta} = (\Delta_1, \dots, \Delta_k) \in \{0,1\}^k$}
      \STATE $\mathbf{D}^\Delta \leftarrow \mathbf{D}$
      \FORALL{$l \in [k], i \in [n]$}
        \IF{$x_{i,j} \neq \star$ and $x_{i, j} \neq \Delta_l$}
          \STATE $\mathbf{D}^\Delta_{i,l} \leftarrow \mathbf{D}^\Delta_{i,l} + 1$
        \ENDIF
      \ENDFOR
      \STATE $\dpmap{j} \leftarrow \dpmap{j} \cup \{\mathbf{D}^\Delta\}$
    \ENDFOR
  \ENDFOR
\ENDFOR
\FORALL{$\mathbf{D} \in \dpmap{m}$}
  \IF{$\min_{l \in [k]} \mathbf{D}_{i,l} \leq \lfloor(1-\tau)m_i\rfloor$ for all $i \in [n]$}
    \RETURN \texttt{TRUE}
    \ENDIF
\ENDFOR
\RETURN \texttt{FALSE}
\end{algorithmic}
\end{algorithm}

\section{Instances with Small Vertex Cover Number}

We adapt the FPT algorithm of \citet{friedrich_binary_2025} for $k$-center with uniform distances, parameterized by the vertex cover number of the \emph{incidence graph}: a bipartite graph where the vertices are the $n$ input strings and $m$ coordinates, connected by an edge if the corresponding entry is present. The incidence graph has a small vertex cover number if all but a few input strings have known entries only within a small subset of positions. For $J \subseteq [m]$, we denote $x_i[J]$ the restriction of $x_i$ to the positions in $J$.

\begin{theorem}
\label{theorem:vc-fpt}
$(k,\tau)$-\textsc{center} can be solved in time
$$
2^{\mathcal{O}\left(\mathrm{vc}(G)^2 \log \mathrm{vc}(G)\right)}\mathrm{poly}(nm),
$$
where $\mathrm{vc}(G)$ is the vertex cover number of the incidence graph $G$.
\end{theorem}
\begin{proof}

\cref{alg:kcenter-vc} adapts the algorithm of \citet{friedrich_binary_2025} for $k$-\textsc{Center with Missing Entries}, which is stated for a single radius $d$ for all input strings. The only change is that instead of reading $d$ from the input, we compute the maximum distance $d_i = \lfloor(1-\tau)m_i\rfloor$ of each string $x_i$ and every subsequent comparison uses $d_i$ in place of $d$. The running time is bounded by
$2^{\mathcal{O}(k \cdot \mathrm{vc}(G) + \mathrm{vc}(G)^2 \log \mathrm{vc}(G))}\mathrm{poly}(nm)$ \citep{friedrich_binary_2025}; as $k$ is a constant, this is $2^{\mathcal{O}(\mathrm{vc}(G)^2 \log \mathrm{vc}(G))}\mathrm{poly}(nm)$.

\begin{algorithm}
  \caption{FPT algorithm for $(k,\tau)$-\textsc{center} parameterized by vertex cover number.}
  \label{alg:kcenter-vc}
  \begin{algorithmic}[1]
    \REQUIRE $x_1, \dots, x_n \in \{0,1,\star\}^m$
    \ENSURE \texttt{TRUE} iff $k$ centers covering $x_1, \dots, x_n$ exist
    \FORALL{$i \in [n]$}
      \STATE $m_i \gets |\{\, j \in [m] \mid x_{i,j} \neq \star \,\}|$
      \STATE $d_i \gets \lfloor (1-\tau) m_i \rfloor$
    \ENDFOR
    \STATE $S \gets$ minimum vertex cover of the incidence graph $G$
    \STATE $R_S \gets \{\, i \in [n] \mid r_i \in S \,\}$
    \STATE $C_S \gets \{\, j \in [m] \mid c_j \in S \,\}$
    \STATE $\overline{C}_S \gets [m] \setminus C_S$
    \STATE $P \gets (R_S \to [k]) \times ([k] \times C_S \to \{0,1\})$
    \FORALL{$(\mathsf{part}, \mathsf{cent}) \in P$ \textbf{and} \textsc{isValid}$(\mathsf{part}, \mathsf{cent})$}
        \STATE $\mathsf{soln} \gets \textbf{true}$
        \FORALL{$l \in [k]$}
          \STATE $R_l \gets \{\, i \in R_S \mid \mathsf{part}(i) = l \,\}$
          \IF{$R_l = \emptyset$}
            \STATE \textbf{continue}
          \ENDIF
          \FORALL{$i \in R_l$}
            \STATE $d'_i \gets d_i - d_H\bigl(x_i[C_S], \mathsf{cent}(l, C_S)\bigr)$
            \STATE $x'_i \gets x_i[\overline{C}_S]$
          \ENDFOR
          \STATE $d' \gets (d'_i)_{i \in R_l}$
          \STATE $x' \gets (x'_i)_{i \in R_l}$
          \STATE $\mathsf{soln} \gets \mathsf{soln}$ \textbf{and} $\textsc{NUCS}(d', x')$
          \IF{\textbf{not} $\mathsf{soln}$}
            \STATE \textbf{break}
          \ENDIF
        \ENDFOR
        \IF{$\mathsf{soln}$}
          \STATE $\overline{R}_S \gets \{\, i \in [n] \setminus R_S \mid x_i[C_S] \text{ has a known entry} \}$
          \STATE $\mathsf{counter} \gets 0$
          \FORALL{$i \in \overline{R}_S$}
            \FOR{$l \in [k]$}
              \IF{$d_H\bigl(x_i[C_S], \mathsf{cent}(l, C_S)\bigr) \le d_i$}
                \STATE $\mathsf{counter} \gets \mathsf{counter} + 1$
                \STATE \textbf{break}
              \ENDIF
            \ENDFOR
          \ENDFOR
          \IF{$\mathsf{counter} = |\overline{R}_S|$}
            \RETURN \texttt{TRUE}
          \ENDIF
        \ENDIF
    \ENDFOR
    \RETURN \texttt{FALSE}
  \end{algorithmic}
\end{algorithm}

Let $S$ be a minimum vertex cover of $G$, let $R_S$ and $C_S$ be the strings and positions it contains, and let $\overline{C}_S = [m] \setminus C_S$. Since $S$ covers all edges of $G$, we have $x_{i,j} = \star$ whenever $i \notin R_S$ and $j \notin C_S$; we call the strings indexed by $R_S$ \emph{long} and the remaining ones \emph{short}, as the latter have known entries only in $C_S$. The algorithm enumerates all \emph{partial assignments} $(\mathsf{part}, \mathsf{cent})$, where $\mathsf{part}$ assigns each long string to one of the $k$ clusters and $\mathsf{cent}$ fixes the values of all $k$ centers on the positions in $C_S$; the test \textsc{isValid} discards the pairs for which some long string $x_i$ is already at distance more than $d_i$ from its assigned center on $C_S$.

Fix a valid pair of $(\mathsf{part}, \mathsf{cent})$. A long string $x_i$ assigned to cluster $l$ may be at distance at most $d'_i = d_i - d_H(x_i[C_S], \mathsf{cent}(l,C_S))$ from the center of $l$ on the remaining positions. Deciding whether cluster $l$ admits such a center is an instance of the \textsc{Non-uniform Closest String} problem on the strings $x_i[\overline{C}_S]$ with radii $d'_i$. \citet{knop_combinatorial_2020} give an algorithm, which we denote \textsc{NUCS}, solving an instance of $|R_l|$ strings of length $|\overline{C}_S|$ in time $|R_l|^{\mathcal{O}(|R_l|^2)}|\overline{C}_S| \in 2^{\mathcal{O}\left(\mathrm{vc}(G)^2 \log \mathrm{vc}(G)\right)}$.

Deciding whether cluster $l$ admits such a center is an instance of the \textsc{Non-uniform Closest String} problem on the strings $x_i[\overline{C}_S]$ with radii $d'_i$. \citet{knop_combinatorial_2020} give an algorithm that solves it in time $n^{\mathcal{O}(n^2)}m$, we denote it \textsc{NUCS}. The second test is for short strings --- with no known entries outside $C_S$ --- so their distance to every center is already determined by $\mathsf{cent}$; they can be assigned to any center within their radius independently of one another. The instance is a yes-instance if and only if some valid pair of $(\mathsf{part}, \mathsf{cent})$ passes both tests. The proof of this equivalence is that of \citet{friedrich_binary_2025} with $d$ replaced by $d_i$ throughout, and we refer to their paper for the details.
\end{proof}

The same adaptation is less immediate for the other two parameters of \citet{friedrich_binary_2025}, the treewidth and the fracture number of the incidence graph. The running time of their treewidth algorithm depends on the radius, which in our setting is proportional to $m_i$. The fracture number has a different obstacle, as their algorithm treats separately the case where the radius is at least twice the fracture number, in which a string with few known entries is covered by every center and can be ignored. Our radii scale with the number of known entries of a string, so this case never occurs and no string comes for free. An FPT algorithm parameterized by the fracture number would be of particular interest for delegated voting, as a small fracture number describes an electorate with a few global issues that all voters care about, where the voters additionally split into small local groups, each concerned only with its own local issues.

\section{Hardness of Approximation}

\newcommand{\conm}[0]{M}

To classify the approximability of \textsc{max}-$(1,\tau)$-\textsc{center}, we use the dichotomy of \citet{khanna_approximability_2001} for the maximization version of constraint satisfaction problems.
 
Given a constraint language $\Gamma$ over the binary domain, the problem $\mathrm{MAX\text{-}CSP}(\Gamma)$ takes a finite set of variables and a finite set of constraints over $\Gamma$ as input, and returns an assignment satisfying as many constraints as possible; an assignment satisfying at least a $1/\rho$ fraction of that maximum is \emph{$\rho$-approximate}. We assume that each constraint on the input has both a satisfying and a non-satisfying assignment. For any $\Gamma$, $\mathrm{MAX\text{-}CSP}(\Gamma)$ is in APX, that is, admits a polynomial-time algorithm that produces a $\rho$-approximate assignment on every input \citep{khanna_approximability_2001} for some constant $\rho$. The dichotomy below determines when $\mathrm{MAX\text{-}CSP}(\Gamma)$ is in PO, the class of optimization problems whose optimum can be computed exactly in polynomial time.
 
The tractable cases are characterized by 2-monotonicity. A constraint $c$ is \emph{2-monotone} if it is expressible in disjunctive normal form with at most two terms, the first containing only positive literals and the second only negative literals, and $\Gamma$ is 2-monotone if each $c \in \Gamma$ is 2-monotone. For $s \in \{0,1\}^k$, let $\mathrm{Zeros}(s) = \{ i \in [k]: s_i=0\}$ and $\mathrm{Ones}(s) = \{ i \in [k]: s_i=1\}$; a set $V \subseteq [k]$ is a \emph{0-term} (\emph{1-term}) of a $k$-ary constraint $c$ if every assignment $s$ with $\mathrm{Zeros}(s) \supseteq V$ ($\mathrm{Ones}(s) \supseteq V$) satisfies $c$.

\begin{lemma}[\citet{khanna_approximability_2001}]
\label{lemma:2-monotone}
A constraint $c$ is 2-monotone iff all of the following conditions are satisfied:
\begin{itemize}
    \item for every satisfying assignment $s$ of $c$, either $\mathrm{Zeros}(s)$ is a 0-term or $\mathrm{Ones}(s)$ is a 1-term,
    \item $V_1 \cap V_2$ is a 1-term for any two 1-terms $V_1$ and $V_2$ of $c$,
    \item $V_1 \cap V_2$ is a 0-term for any two 0-terms $V_1$ and $V_2$ of $c$.
\end{itemize}
\end{lemma}

\begin{theorem}[\citet{khanna_approximability_2001}]
\label{theorem:khanna-max-csp}
$\mathrm{MAX\text{-}CSP}(\Gamma)$ is in PO if any of the following three conditions is satisfied:
\begin{itemize}
    \item $\mathbf{1}$ is a polymorphism of $\Gamma$,
    \item $\mathbf{0}$ is a polymorphism of $\Gamma$,
    \item $\Gamma$ is 2-monotone.
\end{itemize}
Otherwise, $\mathrm{MAX\text{-}CSP}(\Gamma)$ is APX-complete.
\end{theorem}

We use these two results to establish a constant-factor approximation algorithm for a restricted class of instances of \textsc{max}-$(1,\tau)$-\textsc{center}, and APX-hardness of the problem in general. We say that an instance is \emph{$M$-sparse}, for a positive integer $M$, if every input string reveals at most $M$ positions, that is, $m_i \leq M$ for all $i \in [n]$.
 
Recall the notation of binary strings and relations from \cref{definition:string-relation}.

\begin{definition}
Let $0 < \tau < 1$ and $M$ be a positive integer. We define the constraint language $$\Gamma_{\tau, M}^* = \left\{R^{(\tau,i,j)}\right\}_{0\leq j \leq i \leq M}.$$
\end{definition}

\begin{lemma}
\label{lemma:apx-membership}
For any $\tau \in (0, 1)$ and any positive integer $M$, \textsc{max}-$(1,\tau)$-\textsc{center} restricted to $M$-sparse instances is in APX.
\end{lemma}
\begin{proof}
Given an $M$-sparse instance of \textsc{max}-$(1,\tau)$-\textsc{center} with $n$ strings of length $m$, we construct an instance of $\mathrm{MAX\text{-}CSP}(\Gamma^*_{\tau, M})$ with $n$ constraints and $m$ variables $U=\{v_1,\dots,v_m\}$, one for each position. If the string $x_i$ has value $1$ at the positions $j_1,\dots,j_k$ and value $0$ at the positions $j_{k+1},\dots,j_{m_i}$, we add the constraint $R^{(\tau,m_i,k)}(v_{j_1},\dots,v_{j_{m_i}})$.
 
There is a 1-to-1 correspondence between any assignment $f: U \rightarrow \{0,1\}$ of $\mathrm{MAX\text{-}CSP}(\Gamma^*_{\tau, M})$ and the center $y\in\{0,1\}^m$ such that $y_j=f(v_j)$, in the sense that $f$ satisfies exactly the constraints whose corresponding strings are covered by $y$. Since $\Gamma^*_{\tau,M}$ is a finite constraint language, $\mathrm{MAX\text{-}CSP}(\Gamma^*_{\tau, M})$ is in APX \citep{khanna_approximability_2001}, so there is a constant $\rho$ and a polynomial-time algorithm returning a $\rho$-approximate assignment. Its corresponding center $\rho$-approximates the maximum number of input strings that can be covered by a single center.
\end{proof}

For APX-hardness, we apply Khanna's dichotomy theorem to the constraint language $\Gamma_{\tau, M}$ from \cref{def:gamma-tau-m}.

\begin{lemma}
\label{lemma-maxcsp-hard}
$\mathrm{MAX\text{-}CSP}(\Gamma_{\tau, M})$ is APX-complete if $0 < \tau < 1$ and $M \geq 1/\min(\tau,1-\tau)$.
\end{lemma}
\begin{proof}
Using the same argument as in the proof of \cref{lemma-csp-hard}, for $M \geq 1/\tau$, $\mathbf{1}$ is not a polymorphism of $R^\textbf{no}$ and $\mathbf{0}$ is not a polymorphism of $R^\textbf{yes}$.

Next, we apply the second condition of \cref{lemma:2-monotone} to show that the constraint $R^\textbf{yes}$ is not 2-monotone, hence the language $\Gamma_{\tau, M}$ is not 2-monotone, which together with \cref{theorem:khanna-max-csp} implies APX-completeness of $\mathrm{MAX\text{-}CSP}(\Gamma_{\tau, M})$.

Consider the minimum agreement $\mu = \lceil \tau M \rceil$ and assume $M \geq 1/\min(\tau,1-\tau)$, so that $1 \leq \mu \leq M-1$.
Take the two sets $V_1=\{1,\dots,\mu\}$ and $V_2=\{M-\mu+1,\dots,M\}$ of size $\mu$. Any assignment $s \in \{0,1\}^M$ with either $V_1 \subseteq \mathrm{Ones}(s)$ or $V_2 \subseteq \mathrm{Ones}(s)$ has at least $\mu=|V_1|=|V_2|$ ones, hence $s \in R^\textbf{yes}$ and both $V_1$ and $V_2$ are 1-terms by definition. However, $V_1 \cap V_2$ is not a 1-term because $|V_1 \cap V_2|=\max(0,M-2\lfloor (1-\tau)M \rfloor) < \lceil \tau M \rceil$, and the assignment $q$ with $\mathrm{Ones}(q)=V_1 \cap V_2$ is not in $R^\textbf{yes}$.
\end{proof}

\begin{theorem}
\label{theorem:apx-hard}
For any $\tau \in (0, 1)$, \textsc{max}-$(1,\tau)$-\textsc{center} is APX-hard already on $M$-sparse instances, where $M = \lceil 1/\min(\tau,1-\tau)\rceil$ is a constant depending only on $\tau$. In particular, \textsc{max}-$(1,\tau)$-\textsc{center} is APX-hard.
\end{theorem}
\begin{proof}
By \cref{lemma-maxcsp-hard}, $\mathrm{MAX\text{-}CSP}(\Gamma_{\tau, M})$ is APX-hard for this value of $M$. Given a $\mathrm{MAX\text{-}CSP}(\Gamma_{\tau, M})$ instance, we use the same reduction as in the proof of \cref{theorem:delegation-np-complete} to construct an instance of \textsc{max}-$(1,\tau)$-\textsc{center} in which every string reveals exactly $M$ positions, and which is therefore $M$-sparse. Because there is a 1-to-1 correspondence between the satisfied constraints of $\mathrm{MAX\text{-}CSP}(\Gamma_{\tau, M})$ and the strings covered in the reduced instance, for any assignment and its corresponding center, the reduction preserves the objective value exactly. It is therefore a PTAS reduction, and consequently \textsc{max}-$(1,\tau)$-\textsc{center} is APX-hard on $M$-sparse instances. The last claim follows as $M$-sparse instances form a subclass of all instances.
\end{proof}
\end{document}